\documentclass[10pt]{article}
\usepackage{graphicx} %
\graphicspath{{figures/}}
\usepackage{fullpage}
\usepackage{titling}
\usepackage[numbers]{natbib}
\usepackage{amssymb}
\usepackage{xfrac}
\let\nicefrac\sfrac

\usepackage{bm, bbm, xspace, complexity, amsmath, amssymb, physics2, multirow, hhline, makecell, mathtools, nicematrix, mleftright}
\usepackage[table]{xcolor}

\mleftright

\usephysicsmodule{qtext.legacy,ab}

\newcommand{\delimit}[3]{
    \NewDocumentCommand#1{sO{}m}{
    \IfBooleanTF##1
    {\left#2 ##3 \right#3}
    {\mathopen{##2 #2} ##3 \mathclose{##2 #3}}%
    }
}

\delimit\ceil\lceil\rceil
\delimit\floor\lfloor\rfloor
\delimit\ip\langle\rangle
\delimit\norm\lVert\rVert
\delimit\abs\lvert\rvert

\def\va{{\bm{a}}}
\def\vb{{\bm{b}}}

\def\ve{{\bm{e}}}

\def\vm{{\bm{m}}}

\def\vs{{\bm{s}}}
\def\vt{{\bm{t}}}
\def\vu{{\bm{u}}}
\def\vv{{\bm{v}}}
\def\vw{{\bm{w}}}
\def\vx{{\bm{x}}}
\def\vy{{\bm{y}}}
\def\vz{{\bm{z}}}

\renewcommand\vec\bm

\newcommand{\mat}{\mathbf}

\def\mA{{\mathbf{A}}}
\def\mB{{\mathbf{B}}}

\def\mI{{\mathbf{I}}}

\def\mK{{\mathbf{K}}}

\def\mU{{\mathbf{U}}}

\def\mW{{\mathbf{W}}}

\def\cA{{\mathcal{A}}}
\def\cB{{\mathcal{B}}}

\def\cD{{\mathcal{D}}}
\def\cE{{\mathcal{E}}}
\def\cF{{\mathcal{F}}}
\def\cG{{\mathcal{G}}}

\def\cM{{\mathcal{M}}}

\def\cQ{{\mathcal{Q}}}

\def\cS{{\mathcal{S}}}
\def\cT{{\mathcal{T}}}

\def\cX{{\mathcal{X}}}
\def\cY{{\mathcal{Y}}}

\DeclareMathOperator*{\argmax}{argmax}
\let\E\relax
\DeclareMathOperator*{\E}{\mathbb E}

\DeclareMathOperator*{\clip}{clip}

\renewcommand{\R}{\mathbb R}
\newcommand{\N}{\mathbb N}
\newcommand{\Z}{\mathbb Z}

\renewcommand{\S}{\mathbb S}

\let\ds\displaystyle
\let\op\operatorname
\let\eps\epsilon

\newcommand{\unif}{\op{unif}}
\newcommand{\supp}{\op{supp}}
\newcommand{\ind}{\mathbbm{1}}
\newcommand{\Id}{\op{Id}}
\newcommand{\defeq}{\coloneqq}
\newcommand{\eqdef}{\eqqcolon}

\newcommand{\ie}{{\em i.e.}\xspace}
\newcommand{\eg}{{\em e.g.}\xspace}
\newcommand{\conv}{\op{conv}}
\newcommand{\Geo}{\op{Geo}}

\newcommand{\KL}{d_\text{KL}}

\let\grad\nabla

\newcommand{\shellelips}{\textsf{ShellEllipsoid}}
\newcommand{\shellgd}{\textsf{ShellGD}}
\newcommand{\shellproj}{\textsf{ShellProject}}
\newcommand{\enfuns}{\Phi_{\vec{m}}}
\newcommand{\tilY}{\tilde{\cY}}
\newcommand{\tilPhi}{\tilde{\Phi}}
\newcommand{\tilphi}{\tilde{\phi}}

\DeclareMathOperator*{\vol}{vol}

\usepackage{tikz-cd}
\usepackage{forest}
\usetikzlibrary{calc, positioning, external, arrows.meta}

\usepackage[colorlinks, linkcolor=magenta, citecolor=teal]{hyperref}
\usepackage{amsthm}

\usepackage[vlined,ruled]{algorithm2e} %

\SetAlFnt{\small}
\SetAlCapFnt{\small}
\SetAlCapNameFnt{\small}
\SetAlCapHSkip{0pt}
\IncMargin{-\parindent}

\usepackage[capitalize, nameinlink]{cleveref}
\usepackage{autonum}

\usepackage{lineno}

\usepackage{enumitem}

\theoremstyle{plain}
\newtheorem{introtheorem}{Theorem}
\newtheorem{theorem}{Theorem}[section]
\newtheorem*{theorem*}{Theorem}

\newtheorem{claim}[theorem]{Claim}

\newtheorem{lemma}[theorem]{Lemma}

\newtheorem{proposition}[theorem]{Proposition}
\newtheorem{corollary}[theorem]{Corollary}
\newtheorem{question}{Question}
\makeatletter\let\c@question\c@theorem\makeatother
\theoremstyle{definition}
\newtheorem{definition}[theorem]{Definition}
\newtheorem{remark}[theorem]{Remark}
\newtheorem{assumption}[theorem]{Assumption}

\crefname{introtheorem}{theorem}{theorems}
\Crefname{introtheorem}{Theorem}{Theorems}
\crefname{theorem}{theorem}{theorems}
\Crefname{theorem}{Theorem}{Theorems}
\crefname{conjecture}{conjecture}{conjectures}
\Crefname{conjecture}{Conjecture}{Conjectures}
\crefname{claim}{claim}{claims}
\Crefname{claim}{Claim}{Claims}
\crefname{observation}{observation}{observations}
\Crefname{observation}{Observation}{Observations}
\crefname{lemma}{lemma}{lemmas}
\Crefname{lemma}{Lemma}{Lemmas}
\crefname{fact}{fact}{facts}
\Crefname{fact}{Fact}{Facts}
\crefname{proposition}{proposition}{propositions}
\Crefname{proposition}{Proposition}{Propositions}
\crefname{corollary}{corollary}{corollaries}
\Crefname{corollary}{Corollary}{Corollaries}
\crefname{question}{question}{questions}
\Crefname{question}{Question}{Questions}
\crefname{definition}{definition}{definitions}
\Crefname{definition}{Definition}{Definitions}
\crefname{remark}{remark}{remarks}
\Crefname{remark}{Remark}{Remarks}
\crefname{assumption}{assumption}{assumptions}
\Crefname{assumption}{Assumption}{Assumptions}

\usepackage{thm-restate}

\definecolor{briancolor}{rgb}{.2, 0, 1}

\newcommand{\Reg}{\mathsf{Reg}}
\newclass{\UEOPL}{UEOPL}
\newclass{\SSG}{SSG}
\newclass{\CLS}{CLS}
\newclass{\Contr}{Contr}
\newclass{\UnkContr}{UnkContr}

\allowdisplaybreaks
\title{The Complexity of Correlated Equilibria \\Beyond Normal-Form Games}

\usepackage{authblk}

\author[1]{Ioannis Anagnostides}
\author[2]{Constantinos Daskalakis}
\author[2]{Gabriele Farina}
\author[3,4]{\authorcr Noah Golowich}
\author[1,5]{Tuomas Sandholm}
\author[2]{Brian Hu Zhang}

\affil[1]{Carnegie Mellon University}
\affil[2]{Massachusetts Institute of Technology}
\affil[3]{Microsoft Research}
\affil[4]{University of Texas, Austin}
\affil[5]{Additional affiliations: Strategy Robot, Inc., Strategic Machine, Inc., Optimized Markets, Inc.}
\affil[ ]{}
\affil[ ]{\texttt{\{ianagnos,sandholm\}}\texttt{@cs.cmu.edu}, \texttt{\{costis,gfarina,zhangbh\}}\texttt{@mit.edu}, \texttt{nzg@cs.utexas.edu}}

\date\today

\begin{document}

\pagestyle{empty}

\begin{titlepage}
\maketitle

\thispagestyle{empty}

\begin{abstract}
    Correlated equilibria are a fundamental solution concept in game theory, introduced in seminal work by Robert Aumann in the 1970s. A textbook fact is that correlated equilibria enjoy more desirable computational properties \emph{vis-\`a-vis} Nash equilibria: they can be computed in polynomial time in normal-form games via linear programming, and are also amenable to learning via decentralized \emph{no-swap-regret} dynamics. However, despite decades of research, the complexity beyond games of polynomial type---such as extensive-form games, congestion or routing games, and more broadly concave games---has remained a major open problem, first highlighted by Papadimitriou and Roughgarden (JACM '08).
    
    We show that computing an approximate correlated equilibrium of a polytope game is \textsf{PPAD}-complete, and thus polynomial-time equivalent to finding a Nash equilibrium. This challenges the conventional view in game theory that correlated equilibria are more tractable objects than Nash equilibria. From a conceptual standpoint, we connect correlated equilibria to a natural generalization of fixed points, which we call \emph{fixed sets}. A fixed set of a continuous function $f : \mathcal{X} \to \mathcal{X}$ is a nonempty finite set $S \subseteq \mathcal{X}$ such that $S \subseteq \operatorname{conv} f(S)$. Our main technical contribution is to prove that computing a fixed set with a prescribed polynomial support bound is \PPAD-hard even at inverse polynomial precision. Moreover, we establish an unconditional, information-theoretic lower bound ruling out the existence of a strongly sublinear swap regret minimizer: any online learning algorithm requires exponentially many iterations in the dimension $d$ to guarantee at most $1/\poly(d)$ (average) swap regret.

    To circumvent these hardness results, we examine the complexity of $\Phi$-equilibria---tractable relaxations of correlated equilibria. We obtain a fully polynomial-time approximation scheme (FPTAS) for computing poly-dimensional $\Phi$-equilibria in general concave games. We complement this by showing that \Contr-hardness persists even under poly-dimensional swap deviations in the regime where the precision $\epsilon$ is exponentially small. Finally, we show that \Contr-hardness can be bypassed in the canonical setting of concave \emph{quadratic games}, for which we provide a $\poly(d, \log(1/\epsilon))$-time algorithm for computing poly-dimensional $\Phi$-equilibria.
\end{abstract}
\newpage
\thispagestyle{empty}
\tableofcontents
\end{titlepage}

\pagestyle{plain}

\section{Introduction}

\emph{Correlated equilibria} (CEs) were famously put forward in a pathbreaking work of~\citet*{Aumann74:Subjectivity} as a manifestation of Bayesian rationality in games. Unlike \emph{Nash equilibria}~\citep{Nash50:Equilibrium}, the predominant solution concept in noncooperative games, correlated equilibria make allowance for \emph{correlated randomization} between players, typically modeled via a mediator---a trusted third party. As such, correlated equilibria can unlock more desirable outcomes~\citep{Rudov25:Extreme,Ashlagi08:Value}. Moreover, the set of correlated equilibria can be captured via linear constraints, and therefore form a convex polytope. In particular, a correlated equilibrium can be computed in polynomial time in explicitly represented normal-form games, and even in many succinct multi-player games~\citep{Papadimitriou08:Computing}; this stands in stark contrast to Nash equilibria, which are \PPAD-hard to compute even in two-player games~\citep{Daskalakis08:Complexity,Chen09:Settling}. Adding to its plausibility, a correlated equilibrium can be obtained in a decentralized way through the repeated interaction of \emph{no-swap-regret} players~\citep{Blum07:From,Stoltz07:Learning}.

While correlated equilibria can be phrased as solutions to a linear program, the complexity of computing one turns out to be significantly nuanced. Existing algorithms---most notably swap regret minimization~\citep{Blum07:From} and ellipsoid against hope of~\citet*{Papadimitriou08:Computing}---crucially hinge on the game being of \emph{polynomial type}, and in particular that the number of pure strategies of each player is polynomial in the description of the game. This assumption, however, is often violated. Many practical scenarios feature either sequential moves and/or acting upon observing the transitions in the state of the environment, leading to a combinatorial growth in each player's strategy set. Canonical examples include \emph{extensive-form games}~\citep{Shoham08:Multiagent} and \emph{congestion games}~\citep{Rosenthal73:Class}---such as networking routing games~\citep{Roughgarden07:Routing}, \emph{stochastic games}~\citep{Shapley53:Value,Condon92:Complexity}, and \emph{Blotto games}~\citep{Borel21:Theorie}.
Existing linear programming formulations for correlated equilibrium computation in such settings have at least exponentially many variables \emph{and} exponentially many constraints, presenting a daunting algorithmic challenge.

As a result, computing CEs beyond games of polynomial type has remained a major open question in the field, as articulated by~\citet*{Papadimitriou08:Computing} and~\citet*{Stengel08:Extensive}. Meanwhile, the problem is not amenable to the standard toolkit of the $\TFNP$ family of problems for the purposes of establishing computational intractability results. {A key obstacle for establishing intractability results is that the set of CEs is convex; for example, to our knowledge, this immediately rules out all known techniques for showing \PPAD-hardness, as these techniques do not operate with convex solution sets.} All in all, it seems that significantly new techniques should be developed to capture the complexity of the problem. This is the main motivating question for this work.

\begin{question}
    \label{question:CE}
    What is the complexity of computing correlated equilibria beyond games of polynomial type, for example in extensive-form games?
\end{question}

We are interested in studying this question for games of exponential type and beyond. A convenient abstraction for this class is to consider {\em concave games}, which have been of considerable interest since their introduction in the seminal work of~\citet*{Rosen65:Existence}. These model settings in which each player's strategy set is a convex compact subset of Euclidean space, and the utility of each player is concave in that player's strategy when the other players remain fixed (precise definitions are provided in~\Cref{sec:prels}). In particular, in a sharp contrast to the standard case of normal-form games, the strategy set here is infinite (though Nash equilibria still exist~\citep{Rosen65:Existence}). Normal-form games and extensive-form games, for example, are special cases of concave games in which the utilities are  linear, and the strategy sets are simplices and specially-structured polytopes (``tree-form strategy sets'') respectively.

Recent breakthrough results~\citep{Peng24:Fast,Dagan24:From} provided a polynomial-time approximation scheme ($\PTAS$) for computing correlated equilibria in concave games, but the complexity remains wide open beyond the regime where the precision is an absolute constant. Specifically, these prior algorithms are based on no-regret learning algorithms that are \emph{weakly sublinear}, in the sense that their (time-average) regret is just barely vanishing, at a rate of $1/\log T$ for any time horizon $T$. As such, they require $\exp(\Omega(1/\epsilon))$ rounds to guarantee that the (time-average) swap regret of the learner is at most $\epsilon$. An important open question is whether there exists a \emph{strongly sublinear} online algorithm for minimizing swap regret while at the same time incurring a polynomial dependence on the dimension.

\begin{question}
    \label{question:online}
    Is there a strongly sublinear adversarial online learning algorithm, \ie, one with regret $\poly(d) \cdot T^{-c}$ where $c > 0$ is a constant, for minimizing swap regret over  an arbitrary convex compact set $\cX \subset \R^d$, beyond the case of simplices and linear utilities?
\end{question}

A negative answer to this question would imply that correlated equilibria elude computation via no-regret learning, which is a mainstay algorithmic paradigm for equilibrium computation, providing strong evidence for intractability beyond games of polynomial type. On top of the computational complexity implications, \Cref{question:online} is also of interest from an information-theoretic standpoint, and ties to many other problems beyond games in which swap regret plays a crucial role, such as calibration (\Cref{sec:related} elaborates on such connections).

Beyond online learning, we will also be interested in the complexity of computing an equilibrium (by any method, learning or otherwise).

\begin{question}
    \label{question:concave}
    What is the complexity of computing correlated equilibria in concave games? Does the complexity landscape change under concave quadratic utilities?
\end{question}

The apparent inability to compute correlated equilibria beyond games of polynomial type has shifted the research focus to tractable relaxations thereof, primarily stemming from the framework of \emph{$\Phi$-equilibria}~\citep{Greenwald03:General,Stoltz07:Learning}. The set $\Phi$ here comprises strategy deviations; the richer the set $\Phi$, the stronger the induced solution concept. Correlated equilibria correspond to $\Phi$ containing all possible strategy deviations. At the other end of the spectrum, \emph{coarse} correlated equilibria~\citep{Moulin78:Strategically} correspond to $\Phi$ consisting only of constant transformations, and are amenable to usual regret minimization techniques~\citep{Hazan16:Introduction}---even in exponential-type games. The line of work on $\Phi$-regret minimization has flourished in recent years, especially in the context of extensive-form games~\citep{Farina22:Simple,Bai22:Efficient}.

One natural relaxation that lies in between those two extremes is the \emph{linear correlated equilibrium (LCE)}. LCEs are based on the notion of \emph{linear swap regret}, wherein $\Phi$ comprises all linear transformations~\citep{Gordon08:No}, and they form a subset of many other common notions including the {\em extensive-form correlated equilibrium (EFCE)}~\citep{Stengel08:Extensive}. Recent work by~\citet*{Daskalakis25:Efficient}  provided the first efficient algorithms for computing linear correlated equilibria, both in the regret minimization setting and the equilibrium computation setting, positing only oracle access to the underlying strategy sets. In fact, these positive results of~\citet*{Daskalakis25:Efficient} directly apply even in concave games. Approaching closer to swap regret, one can consider a set $\Phi$ with \emph{polynomial dimension}~\citep{Gordon08:No,Zhang25:Learning}; low degree polynomials constitute a canonical such example~\citep{Zhang24:Efficient}. \citet*{Zhang25:Learning} provided both an efficient online learning algorithm and a $\polylog(1/\epsilon)$ equilibrium computation procedure for this broader set of deviations \emph{under (multi)linear utilities}. A key question is whether these positive results can be extended in concave games. We summarize this final research question below.

\begin{question}
    \label{question:polydim}
    When the set $\Phi$ has polynomial dimension, is there a polynomial-time algorithm for computing a $\Phi$-equilibrium in concave games? Is there an efficient learning algorithm for minimizing $\Phi$-regret under concave quadratic utilities?
\end{question}

\subsection{Our results}

We resolve several outstanding open questions concerning the complexity of computing correlated equilibria and swap regret minimization. We gather our main contributions in~\Cref{tab:summary}.

\begin{table}[t]
\definecolor{newcolor}{RGB}{192,255,192}
    \centering
\small
    \renewcommand{\arraystretch}{1.4}

\definecolor{goodcolor}{RGB}{192,255,192}
\definecolor{badcolor}{RGB}{255,192,192}
    \begin{NiceTabular}{c|c|c|c|c|c}[hvlines-except-borders]
    \CodeBefore
        \cellcolor{badcolor}{2-4,4-6}
        \cellcolor{goodcolor}{4-5}
    \Body
    $\Phi$ $\downarrow$ & \bf Game $\rightarrow$ & Extensive-form & Multilin. polytope & \makecell{Poly-dim\\concave utilities} & General concave \\
    \Block{2-1}{Swap (CE)} & EqComp & {\em Open} & \makecell{\PPAD-hard {[Thm~\ref{th:multilinear hardness}]}} & \Block[fill=badcolor]{1-2}{\PPAD-hard [Thm~\ref{cor:concave hardness}]}  \\
    & RegMin & \Block[fill=badcolor]{1-4}{Impossible [Thm~\ref{th:efg-lower}]} \\
    \Block{2-1}{Poly-dim} & EqComp & \Block[]{2-2}{$\P$ \citep{Zhang25:Learning}}&& $\P$ [Thm~\ref{th:concave quadratic eqcomp}] & \Contr-hard [Thm~\ref{th:concave efg hard}]\\
    &  RegMin &  & &  \Block[fill=goodcolor]{1-2}{$\P$ [Thm~\ref{th:rm concave}]} \\
    \Block[]{2-1}{Linear} & EqComp & \Block[]{2-4}{$\P$ \citep{Daskalakis25:Efficient}*} \\
    & RegMin
    \end{NiceTabular}
    \caption{Summary of results. Shaded table cells are new results for this paper; the colors distinguish positive results from negative ones. ``EqComp'' and ``RegMin'' stand for equilibrium computation and adversarial regret minimization respectively. The desired dependence on the dimension $d$ of the strategy sets is polynomial. The desired dependence on the precision $\eps$ is $\polylog(1/\eps)$ for EqComp, and $\poly(1/\eps)$ for RegMin; the support and accuracy bounds for CE hardness are as in~\Cref{cor:concave hardness,th:multilinear hardness}, and poly-dimensional concave utilities are as in~\Cref{th:concave quadratic eqcomp}. *: While they only stated their results for linear utilities, the algorithms of \citet*{Daskalakis25:Efficient} generalize immediately to the concave setting, because they completely circumvent the hardness of convex EFPs by computing actual fixed points of $\phi$, which is easy when $\phi$ is linear.}
    \label{tab:summary}
\end{table}

To cast light on~\Cref{question:CE}, one of our key conceptual contributions is to connect correlated equilibria to a natural relaxation of fixed points, which we call \emph{fixed sets}. A fixed set of $f : \cX \to \cX$ is a nonempty finite set $S \subseteq \cX$ such that $S \subseteq \operatorname{conv} f(S)$ (\Cref{fig:fixed-sets}, right). A fixed set supported on a \emph{single point} is nothing other than a fixed point (\Cref{fig:fixed-sets}, left). Another special case is a \emph{cycle} or periodic orbit under $f$, whereby $S \subseteq f(S)$ (\Cref{fig:fixed-sets}, center). At the heart of our approach is the use of fixed sets as a stepping stone toward characterizing the complexity of correlated equilibria. (\Cref{def:fixed set} gives the approximate version used in our computational results.)

\begin{figure}[htbp]
\centering
\begin{tikzpicture}[
  font=\small,
  >=Stealth,
  point/.style={circle,fill=black,inner sep=1.5pt},
  image point/.style={circle,draw=blue!65!black,fill=white,thick,inner sep=1.5pt},
  map arrow/.style={->,thick,blue!65!black,shorten >=3pt,shorten <=3pt}
]
  \begin{scope}
    \node[] at (2,4.25) {Fixed point};
    \node[point,label=below:{$\vx$}] (point) at (2,1.5) {};
    \draw[map arrow] (point) edge[loop above,min distance=14mm]
      node[above] {$f$} (point);
  \end{scope}

  \begin{scope}[xshift=5.2cm]
    \node[align=center] at (2,4.25)
      {Cycle\\(Periodic orbit)};
    \node[point,label=below:{$\vx_1$}] (cycle1) at (0.6,0.6) {};
    \node[point,label=below:{$\vx_2$}] (cycle2) at (3.4,0.6) {};
    \node[point,label=above:{$\vx_3$}] (cycle3) at (2,2.7) {};
    \draw[map arrow] (cycle1) to[bend right=12] node[below] {$f$} (cycle2);
    \draw[map arrow] (cycle2) to[bend right=12] node[right] {$f$} (cycle3);
    \draw[map arrow] (cycle3) to[bend right=12] node[left] {$f$} (cycle1);
  \end{scope}

  \begin{scope}[xshift=10.4cm]
    \node[align=center] at (2,4.25)
      {Fixed set};
    \coordinate (image1) at (0.3,0.45);
    \coordinate (image2) at (3.7,0.45);
    \coordinate (image3) at (2,2.95);
    \filldraw[fill=blue!6,draw=blue!55!black,dashed]
      (image1) -- (image2) -- (image3) -- cycle;
    \node[point,label=below:{$\vx_1$}] (set1) at (1.4,1.15) {};
    \node[point,label=below:{$\vx_2$}] (set2) at (2.6,1.15) {};
    \node[point,label=below:{$\vx_3$}] (set3) at (2,2) {};
    \draw[map arrow] (set1) -- (image1);
    \draw[map arrow] (set2) -- (image2);
    \draw[map arrow] (set3) -- (image3);
    \node[image point,label=below:{$f(\vx_1)$}] at (image1) {};
    \node[image point,label=below:{$f(\vx_2)$}] at (image2) {};
    \node[image point,label=above:{$f(\vx_3)$}] at (image3) {};
  \end{scope}
\end{tikzpicture}
\caption{A fixed point, a periodic-3 orbit, and a fixed set supported on 3 points.}
\label{fig:fixed-sets}
\end{figure}

One first main technical result establishes that computing a fixed set is as hard as computing a fixed point.

\begin{restatable}{introtheorem}{fixedsethardness}
    \label{th:fixed set hardness}
    For every fixed $C\ge1$, there is a constant $c>0$ such that the following problem is \PPAD-hard: given a Lipschitz function $f:[0,1]^d\to[0,1]^d$ that can be evaluated in polynomial time, compute a $d^{-c}$-fixed set $S$ of $f$ with $|S|\le d^C$.
\end{restatable}

It is worth pointing out that this hardness result, and subsequent implications thereof, extend---under randomized reductions---to distributions with superpolynomial support of size at most $K$, provided the output admits polynomial-time sampling and $\log K$ is polynomial in the input size (\Cref{rem:sampling representations}). Also, the tradeoff between precision and support in~\Cref{th:fixed set hardness} is essential: there is a simple FPTAS for computing an $\eps$-fixed set based on iterating $f$, supported on $O(1/\eps)$ points (\Cref{rem:fixed-set-fptas}).

\paragraph{Proof sketch.} In what follows, for our technical overview, we set aside issues such as approximation. The quantitative version of the argument appears in~\Cref{sec:hardness-fixedsets}.

It is instructive to first treat an intermediate case between fixed points and fixed sets---periodic orbits. In this special case, hardness follows directly from a scaling argument, which has the effect of \emph{amplifying the length of any periodic orbit}. In particular, given an $L$-Lipschitz map $f : [0,1]^d \to [0,1]^d$, we define $f_\alpha(\vx)=(1-\alpha)\vx+\alpha f(\vx)$ for a small parameter $\alpha \in (0,1]$. If $\vx_1,\dots,\vx_k$ form a periodic orbit of $f_\alpha$, with $\vx_{k+1}=\vx_1$, then consecutive points are at distance at most $\alpha$ in the normalized Euclidean norm, so the orbit has diameter at most $k\alpha$. Moreover, summing around the orbit gives $\sum_{i=1}^k (f(\vx_i)-\vx_i)=0$, and hence
\begin{align}
    \norm{f(\vx_1)-\vx_1}_2
    &\le \frac1k\sum_{i=1}^k \norm{(f(\vx_1)-\vx_1)-(f(\vx_i)-\vx_i)}_2
    \le (L+1)k\alpha.
\end{align}
Thus, choosing $\alpha\le\delta/((L+1)k)$ ensures that any orbit of length at most $k$ yields a $\delta$-fixed point of $f$.

However, the challenge for a general fixed set is that each point is a convex combination of images of other points, rather than the image of a single predecessor. Instead of using a single scale, the key idea is to use \emph{several scales} simultaneously: the same transition matrix must form a fixed set at every scale, and we exploit this constraint to recover an approximate fixed point.

Specifically, we construct a function $g : [0, 1]^{dr} \to [0, 1]^{dr}$, where $r$ will be picked later, as follows: we pick $\alpha_1, \dots, \alpha_r \in (0,1]$, and set
\begin{align}
    g(\vx^{(1)}, \dots, \vx^{(r)}) = \ab(f_{\alpha_1}(\vx^{(1)}), \dots, f_{\alpha_r}(\vx^{(r)})).
\end{align}
(This lifting is illustrated in~\Cref{fig:fixed-sets-scales}.) Let $S = \{ \vx_1, \dots, \vx_k \}$ be a fixed set of $g$, and let $P$ be a transition matrix satisfying $\vx_i = \E_{j \sim P(i)} g(\vx_j)$ for every $i \in [k]$. Our goal is to show that $S$ yields an approximate fixed point of $f$; indeed, we will show that there are $i \in [k]$ and $\ell \in [r]$ for which $\vx^{(\ell)}_i$ is an approximate fixed point of $f$.
\begin{figure}[!htb]
\centering
\includegraphics[width=\textwidth]{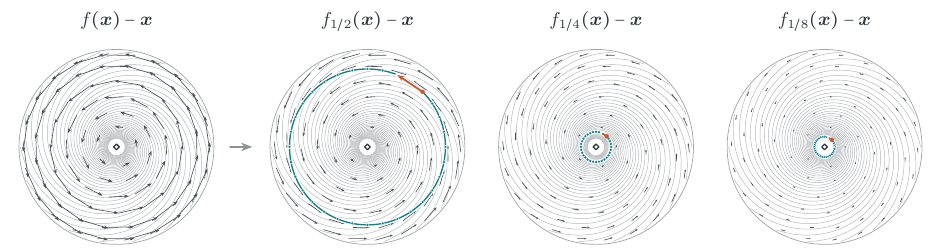}\par
\nointerlineskip
\vspace{4pt}
{\color{black!20}\rule{0.96\textwidth}{0.3pt}}\par
\nointerlineskip
\vspace{4pt}
\includegraphics[width=\textwidth]{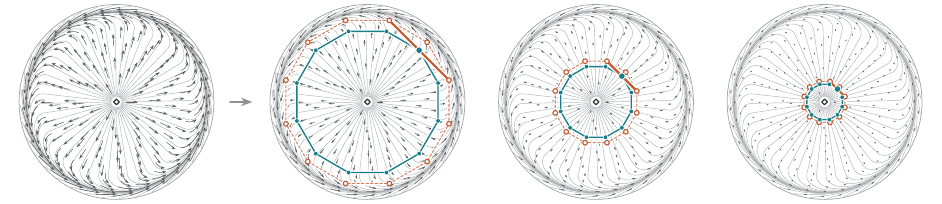}
\caption[Fixed sets at three scales of two different maps.]{Top row. The original rotation $f$ and three scaled maps $f_\alpha$, with common fixed point $\vx_*=0$. Gray curves depict the displacement flow, and black arrows show $f_\alpha(\vx)-\vx$. Teal points form approximate fixed sets with the same transition matrix and equal error across the three scales.
\par\smallskip
Bottom row. A different map $f$ and its three scaled maps. Teal points form regular $12$-gons; hollow orange points are their images under $f_\alpha$. Each support point is the midpoint of the images of its two neighbors, as illustrated by the orange chord. Thus, $P(i,i-1)=P(i,i+1)=1/2$ works at all three scales, with uniform stationary distribution.}
\label{fig:fixed-sets-scales}
\end{figure}

The key observation is the following. For $\ell \in [r]$ let $S_\ell = \{ \vx_1^{(\ell)}, \dots, \vx_k^{(\ell)} \}$. Then, by definition, $S_\ell$ is a fixed set of $f_{\alpha_\ell}$ with transition matrix $P$. Therefore, we have constructed fixed sets at multiple ``scales'' of $f$, {\em that use the same transition matrix}. We show that if we do this for ``enough'' values of $\alpha \in (0,1]$, then we must find a fixed point.

\begin{proposition}\label{prop:non-quantitative}
    Let $P : [k] \to \Delta(k)$ be a transition matrix. For every $\alpha \in (0,1]$ let $S_\alpha$ be a fixed set of $f_\alpha$ with transition matrix $P$. Then the closure of $\bigcup_{\alpha \in (0,1]} S_\alpha$ contains a fixed point of $f$.
\end{proposition}
\begin{proof}
    Assume WLOG that the transition graph of $P$ is strongly connected; otherwise, replace it with a sink strongly connected component. Let $S_\alpha = \{ \vx_1^{(\alpha)}, \dots, \vx_k^{(\alpha)} \}$. Then we have
    \begin{align}\label{eq:non-quantitative}
        \vx_i^{(\alpha)} = \E_{j \sim P(i)} \ab[ (1 - \alpha) \vx_j^{(\alpha)} + \alpha f(\vx_j^{(\alpha)})].
    \end{align}
    Let $\{\alpha_r\}_{r=1}^\infty$ be a sequence converging to $0$, and $\vx_i = \lim_{r \to \infty} \vx_i^{(\alpha_r)}$. (If this limit does not exist, replace $\{\alpha_r\}$ by a subsequence so that it does; this is allowed because of compactness.) Replacing this limit in \eqref{eq:non-quantitative} yields
    \begin{align}
        \vx_i = \E_{j \sim P(i)} \vx_j.
    \end{align}
    We claim now that $\vx_1 = \dots= \vx_k$. Indeed, consider any direction $\va$ and let $i$ be chosen to maximize $\ip{\va, \vx_i}$. Then
    \begin{align}
        \ip{\va, \vx_i} = \E_{j \sim P(i)} \ip{\va, \vx_j}.
    \end{align}
    By induction and strong connectivity of the transition graph, it follows that $\ip{\va, \vx_j} = \ip{\va, \vx_i}$ for {\em every} $j$. Since this is true for every $\va$, we have  $\vx_1 = \dots= \vx_k \eqdef \vx$.

    Moreover, taking expectations in \eqref{eq:non-quantitative} over the stationary distribution $\mu$, and dividing by $\alpha > 0$, yields
    \begin{align}\label{eq:qualitative-efp}
        \E_{i \sim \mu} \vx_i^{(\alpha)} = \E_{i \sim \mu}f(\vx_i^{(\alpha)}).
    \end{align}
    Since $\vx_i^{(\alpha_r)} \to \vx$ for all $i$, taking the limit $r \to \infty$ now yields $\vx = f(\vx)$.
\end{proof}

As alluded to, the connection we make between fixed sets and correlated equilibria yields the following direct consequence for quadratic concave games.

\begin{introtheorem}\label{cor:concave hardness}
    For every fixed $C\ge1$, there is a constant $c>0$ such that computing a $d^{-c}$-correlated equilibrium supported on at most $d^C$ strategy profiles is \PPAD-hard, even in two-player concave games with $\cX_1 = \cX_2 = [0, 1]^d$ where each utility is quadratic in the player's own strategy.
\end{introtheorem}

That is, surprisingly, computing a correlated equilibrium is as hard as computing a Nash equilibrium, challenging the conventional view in game theory. \Cref{cor:concave hardness} follows from~\Cref{th:fixed set hardness} using the standard two-player \emph{imitation game} (\emph{e.g.},~\citealp{Mclennan10:Simple,Babichenko16:Query,Rubinstein16:Settling}), with utilities $u_1(\vx,\vy)=-\norm{\vx-\vy}_2^2$ and $u_2(\vx,\vy)=-\norm{\vy-f(\vx)}_2^2$ (\Cref{prop:ce-fixed-set reduction}).

Moreover, the same reduction readily gives an unconditional lower bound in the query model.

\begin{introtheorem}\label{cor:concave query hardness}
    For every fixed $C\ge 1$, there is a constant $c>0$ such that any randomized algorithm that, with probability at least $2/3$, computes a $d^{-c}$-correlated equilibrium supported on at most $d^C$ profiles requires $2^{\Omega(d/\log d)}$ utility queries in the worst case. This holds even in two-player concave games on $[0,1]^d$, where each utility is quadratic in the player's own strategy.
\end{introtheorem}

To see this, we start from the $2^{\Omega(d_0)}$ randomized query lower bound for finding a fixed point of a Lipschitz map on $[0,1]^{d_0}$, which holds even when the Lipschitz and precision parameters are absolute constants~\citep{Hirsch89:Exponential,Rubinstein16:Settling}. For a prescribed polynomial support bound, our reduction constructs a game of dimension $d=O(d_0\log d_0)$ in which each utility query can be simulated with $O(\log d_0)$ queries to the original map.

Finally, we obtain hardness for multilinear games on explicitly represented polytopes.
\begin{restatable}{introtheorem}{polytopemult}
\label{th:multilinear hardness}
    Let $N$ be the binary encoding length of the game's description. For every fixed $C \ge 1$, computing an $\eps$-correlated equilibrium supported on at most $N^C$ profiles in a multilinear game on explicit rational polytopes is
    \begin{enumerate}[label=(\roman*)]
        \item \PPAD-hard at accuracy $\eps=2^{-\operatorname{poly}(N)}$;
        \item \PPAD-hard at accuracy $\eps=N^{-c}$ for some constant $c>0$, assuming the PCP-for-PPAD conjecture (\Cref{ass:pcp-ppad}).\footnote{OpenAI recently claimed a proof of the required \href{https://github.com/openai/math/blob/main/preprints/The-PCP-for-PPAD-conjecture-a-quasilinear-reduction-September-25-2026/paper.pdf}{PCP-for-PPAD theorem}. We refrain from using this claim as an established theorem until it has been verified by human experts, and state this part conditionally.}
    \end{enumerate}
    Both statements hold even when the number of players is an absolute constant.
\end{restatable}

The obvious approach that comes to mind here is to employ some form of discretization. However, the issue is that this would make the game description grow polynomially with the inverse accuracy. At the same time, \Cref{th:fixed set hardness} prescribes that the accuracy needs to be small enough relative to the size of the support. As a result, it is not clear whether there is a discretization that can handle any polynomial support bound.

Instead, to arrive at~\Cref{th:multilinear hardness}, the key step is to essentially transfer the quadratic imitation term from the utility into the strategy constraints. Indeed, since $-\norm{\vx-\vy}_2^2 = \frac{2}{d}\ip{\vx,\vy}-\norm{\vx}_2^2-\norm{\vy}_2^2$ (in relative $\ell_2$ norm), it suffices to incorporate an auxiliary variable whose value approximates $\norm{\vx}_2^2$, for $\vx$; the last term does not depend on $\vx$, while the interaction term is already bilinear. The challenge is to express these auxiliary variables with few linear constraints. To do so, we rely on an approach in approximation theory (\emph{e.g.,}~\citealp{Yarotsky17:Error}), whereby each term $x_i^2$ in $\|\vx\|_2^2$ is replaced by a piecewise-linear interpolation that can be described with only $O(m)$ auxiliary variables and inequalities.

To cast the remaining polynomial interactions into multilinear form, we use the standard trick of assigning different factors of each monomial to different ``copy players.'' In particular, we can start with a polynomial map $f:[0,1]^d \to[0,1]^d$ of constant degree whose approximate fixed points are \PPAD-hard to compute. Since $f$ has constant degree, we require only a constant number of players. These steps are described at greater length in~\Cref{sec:polytope}.

It is interesting to point out that our results immediately also imply \PPAD-hardness for computing a point in the convex hull of Nash equilibrium distributions---which is another problem we have been curious about---under the same output representation and support-accuracy tradeoff. Indeed, every convex combination of Nash equilibrium distributions is a correlated equilibrium.

\paragraph{Swap regret minimization.} Next, we turn to the online learning setting. We show that there is no strongly sublinear learner for minimizing swap regret, resolving~\Cref{question:online}.

\begin{introtheorem}[Informal version of \Cref{th:efg-lower}]
    \label{theorem:swaponline}
    In extensive-form games, for $\epsilon = \poly(1/d)$, any learner that guarantees at most $\epsilon$ expected swap regret requires at least $\exp(\poly(d))$ iterations. This holds even under a linear sequence of utilities and an oblivious adversary.
\end{introtheorem}

This lower bound is unconditional and information-theoretic: it applies even if the learner is computationally unbounded. The first implication of~\Cref{theorem:swaponline} is that independent no-regret learning, a standard paradigm for computing equilibria in games, inherently fails to converge to correlated equilibria in polynomial time, providing further evidence for intractability even under (multi)linear utilities. But~\Cref{theorem:swaponline} has implications beyond the computation of correlated equilibria, in light of the many applications of swap regret (\Cref{sec:related}). For example, \citet*{Fishelson25:High} recently used~\Cref{theorem:swaponline},\footnote{A short manuscript containing \Cref{theorem:swaponline} has appeared on arXiv as a preliminary unpublished preprint titled ``A Lower Bound on Swap Regret in Extensive-Form Games.''} in conjunction with the well-known reduction of~\citet*{Foster97:Calibrated}, to directly obtain an exponential lower bound on high-dimensional calibration, improving over a quasipolynomial lower bound of~\citet*{Peng25:High}.

\paragraph{Hardness of $\Phi$-equilibria in concave games.}
Our remaining results are about $\Phi$-equilibria. First, we show that computing $\Phi$-equilibria in concave games is $\Contr$-hard.
\begin{introtheorem}[Informal version of \Cref{th:concave efg hard}]
    \label{cor:contr-polydim}
    Computing a $\Phi$-equilibrium of a two-player concave game with respect to poly-dimensional swap deviations is \Contr-hard.
\end{introtheorem}
The hardness holds even when $\Phi$ contains only one deviation function for each player. As shown recently~\citep{Zhang25:Learning}, there is a polynomial-time algorithm for computing a $\Phi$-equilibrium with respect to poly-dimensional swap deviations in multilinear games. As such, \Cref{cor:contr-polydim} marks a rare setting where, unless there is a breakthrough algorithmic advance, namely a polynomial-time algorithm for $\Contr$, computing an equilibrium in concave games is strictly harder than in multilinear games.

\paragraph{A polynomial-time algorithm for $\Phi$-equilibria in quadratic concave games.} Our next result treats the important special case of concave games with quadratic utilities.  We show that the quadratic structure in the utilities enables us to bypass the \Contr-hardness established earlier for general concave games.

\begin{introtheorem}[Informal version of \Cref{th:concave quadratic eqcomp}]
    \label{theorem:quadratic-polytime}
    There is a polynomial-time algorithm for computing $\Phi$-equilibria with respect to poly-dimensional swap deviations in any concave quadratic game when the number of players is constant.
\end{introtheorem}

To put this into context, a well-known fact is that a fixed point of a contraction mapping with respect to the $\ell_2$ norm can be computed in polynomial time~\citep{Sikorski93:Ellipsoid}. One can view~\Cref{theorem:quadratic-polytime} as the analog of that result in the realm of $\Phi$-equilibria. In addition, \Cref{theorem:quadratic-polytime} separates $\Phi$-equilibria---with respect to poly-dimensional swap deviations---from correlated equilibria, as the latter is $\PPAD$-hard even under concave quadratic utilities (\Cref{cor:concave hardness}).

For a class $\cG$ of operators $G : \cX \to \R^d$, an $\eps$-approximate \emph{$\cG$-expected fixed point} ($\cG$-EFP) of a map $\phi : \cX \to \cX$ is a distribution $\mu \in \Delta(\cX)$ satisfying
\begin{align}
    \E_{\vx\sim\mu} \ip{G(\vx), \vx-\phi(\vx)} \le \eps \qquad \text{for every } G\in\cG.
\end{align}

By using the $\cG$-EFP oracle of~\citet*{Daskalakis27:Polynomial} (\Cref{sec:qefp}) as a drop-in replacement for the expected fixed point oracle in the framework of \citet*{Zhang25:Learning}, it follows that, more generally, $\Phi$-equilibria can be computed efficiently if the {\em utilities} are efficiently representable.

\paragraph{An FPTAS for $\Phi$-equilibria.} We complement~\Cref{cor:contr-polydim,theorem:quadratic-polytime} by showing that there is an $\FPTAS$ for computing an $\epsilon$-approximate $\Phi$-equilibrium with respect to poly-dimensional swap deviations in general concave games; to put this into context, we recall that there is an $\FPTAS$ for computing fixed points of non-expansive mappings~\citep{Lieder21:Convergence,Halpern67:Fixed}.

\begin{introtheorem}[Informal version of \Cref{th:concave fptas eqcomp}]
    \label{theorem:FPTAS-polydim}
    There is an $\FPTAS$ for computing an $\epsilon$-approximate $\Phi$-equilibrium of any concave game with respect to poly-dimensional swap deviations when the number of players is constant.
\end{introtheorem}

The first ingredient for establishing this result is an $\FPTAS$ for computing convex expected fixed points (\Cref{def:cefp}). We observe that there is a simple $\FPTAS$ for computing a convex expected fixed point: iterate the mapping $\phi$ multiple times and then take the uniform distribution over that sequence of points. Moreover, this can be used in conjunction with the framework of~\citet*{Zhang25:Learning} for minimizing poly-dimensional swap regret, leading to~\Cref{theorem:FPTAS-polydim}.

\subsection{Further related work}
\label{sec:related}

We now elaborate on related threads of research.

\paragraph{Computing correlated equilibria beyond polynomial type.} For games of polynomial type, as discussed, there are polynomial-time algorithms computing exact CE~\citep{Papadimitriou08:Computing,Jiang15:Polynomial-Time}. A complementary positive result was obtained by~\citet*{Anagnostides26:Normal}, who showed that CEs can be computed efficiently in finite-horizon Markov games with a fixed number of players. Related hardness results concern \emph{sparse} coarse correlated equilibria (CCE), represented as mixtures of a small number of product distributions: \citet*{Foster23:Hardness} established hardness in Markov games, yielding lower bounds for independent learning. Analogous hardness results were later established in extensive-form games~\citep{Anagnostides24:Complexity,Peng24:Complexity}. Unlike these latter two results for extensive-form games, which apply even to coarse correlated equilibria, our main result rules out efficient computation for \emph{any} prescribed polynomial support bound. The complexity of computing an \emph{optimal} correlated equilibrium---as opposed to computing \emph{any} correlated equilibrium, which is our focus in this paper---in extensive-form games was recently characterized by~\citet*{Cheval26:Complexity}.

\paragraph{Tractable relaxations.} One fruitful research direction has been to obtain tractable equilibrium concepts by restricting the set of deviations, in extensive-form games~\citep{Stengel08:Extensive} and beyond. \citet*{Farina24:Polynomial} gave a polynomial-time algorithm for \emph{linear $\Phi$-equilibria}---whereby only linear deviations are considered---in polyhedral games, while~\citet*{Daskalakis25:Efficient} established efficient learning and computation in general convex games. \citet*{Arunachaleswaran25:Swap} introduced \emph{profile swap regret} for polytope games and gave efficient algorithms guaranteeing sublinear profile swap regret, which they showed is necessary and sufficient for \emph{non-manipulability}. This extends the non-manipulability guarantee of \emph{swap regret} in normal-form games~\citep{Deng19:Strategizing}, where a strategic opponent cannot asymptotically obtain more than its Stackelberg payoff against a learner with sublinear swap regret; we recall that sublinear swap regret leads to correlated equilibria~\citep{Hart00:Simple,Blum07:From}. \citet*{Zhang25:Learning} considered a class of nonlinear deviations, including low-degree polynomial deviations. These developments build on the connection between fixed-point computation and regret minimization~\citep{Hazan07:Computational,Gordon08:No}. A central ingredient in their approach is the notion of an \emph{expected} fixed point, which can be computed efficiently~\citep{Zhang25:Learning}. Our fixed sets correspond instead to \emph{convex} expected fixed points, revealing a computational separation between these two relaxations of fixed points.

More broadly, equilibrium computation in concave games typically shares the same complexity as multilinear games, which makes our hardness result in~\Cref{thm:polydimcontract} more surprising. Motivated by recent applications in machine learning, there has been interest in extending the scope to non-concave games, which have been referred to as the next frontier for algorithmic game theory research~\citep{Daskalakis21:Non}. Identifying tractable notions of correlated equilibrium in non-concave games has been explored in recent research~\citep{Zhang25:Expected,Cai24:Tractable,Cai26:Proximal,Ahunbay25:First}.

\paragraph{Learning with large strategy spaces.}
The notion of swap regret was introduced by~\citet*{Blum07:From}, and is a strengthening of \emph{internal regret}~\citep{Hart00:Simple,Stoltz05:Internal}. On the probability simplex, swap and internal regret are polynomially related~\citep{Blum07:From}, but in games of exponential type employing internal regret is in a certain sense vacuous in that the uniform random strategy guarantees exponentially small internal regret (\emph{e.g.}, we refer to~\citet*{Fujii23:Bayes}).

\citet*{Stoltz07:Learning} examined convergence to correlated equilibria in games with convex strategy sets, and~\citet*{Stein07:Characterization} developed approximation algorithms for correlated equilibria in continuous games, including polynomial games. More recently,~\citet*{Peng24:Fast,Dagan24:From} obtained algorithms for minimizing swap regret even in large strategy spaces, yielding a polynomial-time approximation scheme for correlated equilibria in, for example, extensive-form games. However, our learning hardness result (\Cref{theorem:swaponline}) rules out swap-regret guarantees polynomial in the game size and inverse accuracy. Such lower bounds leave open the possibility of efficient equilibrium computation. On the technical side, our entropy argument is related in spirit to the information-theoretic approach of~\citet*{Harrow16:Tight} for approximating Nash equilibria.

\paragraph{Calibration and economic motivation.} Correlated equilibria are also intimately connected to forecasting and \emph{calibration}. \citet*{Foster97:Calibrated} showed that best responses to calibrated forecasts of opponents' play lead the empirical distribution toward correlated equilibria. This connection has been further cultivated in recent years, where improved algorithms for minimizing swap regret~\citep{Dagan24:From,Peng24:Fast} have led to efficient algorithms for high-dimensional calibration~\citep{Peng25:Calibration,Fishelson25:Calibration}. Swap regret has also been employed in the context of omniprediction and multicalibration to ensure multi-group fairness~\citep{Lu25:Sample,Gopalan23:Swap}. Another use of swap regret is in the context of \emph{algorithmic collusion}~\citep{Calvano20:Artificial}, where it can be used to mitigate supracompetitive pricing~\citep{Hartline24:Regulation,Hartline25:Regulation}. Indeed, vanishing external regret alone can still permit collusive behavior. Recently, \citet*{Hartline26:Computable} introduced the notion of \emph{information-value-free correlated equilibria}, and showed that they can be computed efficiently offline but cannot be learned by a broad class of learning algorithms. These connections show the importance of understanding the computational complexity of equilibrium concepts stronger than coarse correlated equilibria, which is what motivates our paper.

\paragraph{Complexity of computing fixed points of contractions.} The complexity class \emph{unique end of potential line} $(\UEOPL)$ was introduced by~\citet*{Fearnley20:Unique} to capture problems in $\PPAD \cap \PLS$ that are guaranteed to have unique solutions, such as finding a fixed point of a contraction mapping. It is currently not known if $\Contr$ is complete for $\UEOPL$, although~\citet*{Fearnley20:Unique} established completeness for a certain discrete variant of contraction. It is an interesting question whether our hardness result for $\Phi$-equilibria can be strengthened to $\UEOPL$-hardness. The recent breakthrough result of~\citet*{Chen25:Computing} (see also~\citet*{Chen26:Quadratic}) gave an algorithm for computing the fixed point of a contraction mapping that only requires a polynomial number of \textit{queries} (but exponential time to compute each query point). \citet*{Daskalakis18:Converse} proved \CLS-hardness but for a more general \textit{metric} contraction problem where the metric is not necessarily induced by any norm. It remains unclear whether one can employ the reduction of~\citet*{Daskalakis18:Converse} to strengthen our \Contr-hardness.

Relatedly, \citet*{Mehta14:Constant} showed that computing a Nash equilibrium even under the promise that the set of Nash equilibria is convex is \SSG-hard. Also, computing a symmetric Nash equilibrium under the promise that there is a unique symmetric Nash equilibrium is \SSG-hard~\citep{Mehta14:Constant}. These lower bounds strongly tie to the problem of showing hardness for correlated equilibria: any reduction based on some notion of equilibrium collapse between Nash and CEs has to grapple with the fact that the set of CEs is convex.

\paragraph{Fixed points and $\Phi$-equilibria.} In the online learning setting, \citet*{Hazan07:Computational} established an equivalence between computing fixed points and minimizing $\Phi$-regret. As observed by~\citet*{Zhang24:Efficient}, their reduction can be bypassed if the learner can employ \emph{mixed} strategies; without the use of mixed strategies, \citet*{Zhang24:Efficient} used the equivalence shown by~\citet*{Hazan07:Computational} to preclude the existence of computationally efficient algorithms minimizing poly-dimensional swap regret even in Bayesian games. When the underlying transformation is a contraction, one of our key results establishes an equivalence \emph{\`a la}~\citet*{Hazan07:Computational} in the equilibrium computation setting, which was hitherto less understood.

\section{Preliminaries}
\label{sec:prels}

\paragraph{Notation} $\S^{d}_+$ is the semidefinite cone in dimension $d$. For $\mA, \mB \in \R^{d\times d}$, $\mA \succeq \mB$ means $\mA - \mB \in \S^{d}_+$. An {\em endomorphism} on $\cX$ is a function mapping $\cX \to \cX$. For a set $\cX$ and number $\eps \ge 0$, $B(\cX, \eps)$ is the set of all points $\eps$-close in the $\ell_2$-norm to $\cX$. For a set $\cX$, we use $\Delta(\cX)$ to denote the set of {\em finite-support} distributions on $\cX$. If $\mu \in \Delta(\cX)$ and $\nu \in \Delta(\cY)$, then $\mu \times \nu \in \Delta(\cX \times \cY)$ denotes the product distribution. %

In the technical overview and~\Cref{sec:fixed-sets,sec:polytope}, following~\citet*{Rubinstein16:Settling}, for $\vx\in\R^d$ and $1\le p<\infty$ we adopt the normalized norm convention
\begin{align}
    \norm{\vx}_p = \ab(\frac{1}{d} \sum_{i=1}^d |x_i|^p)^{1/p}.
\end{align}
Inner products are defined as usual without any normalization: $\ip{\vx,\vy}=\sum_{i=1}^d x_i y_i$, so $|\ip{\vx,\vy}|/d\le\norm{\vx}_2\norm{\vy}_2$. Elsewhere, $\norm{\cdot}_2$ denotes the unnormalized Euclidean norm.

\subsection{Computational model and representation}\label{sec:computation model}

For computational hardness results, we work in the standard Turing model. All numbers are rational numbers encoded in binary. Functions are given as arithmetic circuits with standard gates, see \eg,~\citet*{Rubinstein16:Settling}. These choices are not crucial to our results; indeed, our results will follow from black-box reductions from Brouwer fixed points to our problems of interest.

For our algorithmic and query-complexity results, we will operate in a black-box model of computation. Basic arithmetic operations on real numbers will be assumed to run in unit time. To avoid issues surrounding numerical precision, we will always be interested in $\eps$-approximate solutions to problems, where $\eps > 0$ is a given input parameter.
Functions will generally be given by evaluation oracles. If $f : \cX \to \R$ is also guaranteed to be convex or concave, we also assume access to an oracle for a subgradient $\grad f : \cX \to \R^d$, and we assume that $\norm{\grad f(\vx)}_2 \le 1$ for all $\vx$, that is, that $f$ is $1$-Lipschitz. 
 We will view the oracle correctness and convexity guarantees (\eg, that $\grad f$ is actually the gradient of $f$; that $f$ is actually convex; that $f$ is Lipschitz and bounded, etc.) as {\em promises}: our positive results will assume that the promise holds, and our negative results will only contain instances in which the promise holds.
 
For our positive results, convex compact sets $\cX$ will be represented by separation oracles, that is, Turing machines that, given $\vx \in \R^d$, either decides that $\vx \in \cX$, or outputs a hyperplane $\va^\top \vx = b$ that strictly separates $\vx$ from $\cX$. We will assume that $B(\vec 0, 1) \subseteq \cX \subseteq B(\vec 0, R)$ where $R > 0$ is given. Together with the separation oracle, these conditions imply the ability to optimize any Lipschitz convex function $f$ (given by an oracle as per the previous paragraph) to precision $\eps$ over $\cX$ in time $\poly(d, \log(R/\eps))$~\citep{Grotschel12:Geometric}.

Our results are not sensitive to these choices; we specify them only for the sake of formality and cleanliness. For example, all positive results in our paper would still hold in reasonable finite-precision Turing machine-based models, or---after standard minor adjustments---with weak separation oracles.

Various problems that we will study, including the $\Phi$-equilibrium problem, require reasoning about distributions over infinite sets $\cX$. To avoid measure-theoretic issues and to ensure that outputs can be represented efficiently by an algorithm, in these cases we will consider only finite-support distributions $\mu$, represented by listing the elements of the support set and their probabilities. Algorithms that output distributions must output them in this form. 
Lower bounds will not be sensitive to this choice, either: our fixed-set hardness results also hold for polynomial-time sampling representations under the support and accuracy bounds of~\Cref{rem:sampling representations} (as long as one permits randomized reductions).

\subsection{No-regret learning} No-regret learning is a standard paradigm for learning in games and for computing various kinds of correlated equilibria~\citep{Cesa-Bianchi06:Prediction}. Here, we give a very general setup for no-regret learning that will encompass all the notions that we care about for this paper. 

A learner has a convex compact strategy set $\cX \subset \R^d$. On each of $T$ rounds $t = 1, \dots, T$, the learner outputs a {\em mixed} strategy $\mu^{(t)} \in \Delta(\cX)$, and an {\em adversary}, after observing $\mu^{(t)}$, selects a concave function $u^{(t)} : \cX \to \R$, which the learner observes. The goal of the learner is to minimize the regret with respect to a given set of {\em strategy transformations} $\Phi \subseteq \cX^\cX$:
\begin{align}
    \Reg_\Phi(T) := \sup_{\phi \in \Phi} \frac{1}{T} \sum_{t=1}^T \ab[ u^{(t)}(\phi(\vx^{(t)})) - u^{(t)}(\vx^{(t)})]. \label{eq:regret}
\end{align}
Readers familiar with no-regret learning may wonder why we care about concave utilities and mixed strategies, since in the usual no-regret paradigm it is often the case that one can assume without loss of generality that utilities are linear and strategies are points in $\cX$. The answer is that, while the above are WLOG when $\Phi$ contains only linear functions, neither assumption is WLOG when $\Phi$ contains {\em nonlinear} functions. We will elaborate on this in the technical sections when it becomes relevant.

A few special cases of $\Phi$ will be of particular interest. If $\Phi$ is the set of linear endomorphisms mapping $\cX$ to itself, the resulting notion is {\em linear regret}~\citep{Gordon08:No}. If $\Phi$ consists of {\em all} functions $\cX \to \cX$, the resulting notion is {\em swap regret}. Finally, following \citet*{Zhang25:Learning}, we will say that $\Phi$ is {\em $k$-dimensional} if it can be expressed as a linear combination of a small number of features, that is, if there is a feature map $\vm : \cX \to \R^k$ for which $k = \poly(d)$ and every map $\phi \in\Phi$ can be written in the form $\phi(\vx) = \mK_{\phi}^\top \vm(\vx)$ for some matrix $\mK_{\phi} \in \R^{d \times k}$. For example, finite sets $\Phi$ are $k$-dimensional with $k = d \cdot |\Phi|$ by  setting $\vm(\vx)$ to be the concatenation of the finitely many possible values of $\phi(\vx)$, and  the set of degree-$\ell$ polynomials is $k$-dimensional with $k = d^{O(\ell)}$, by setting $\vm(\vx)$ to be the vector of up-to-$\ell$th moments of $\vx$. We will always assume, without loss of generality for our settings, that $\Phi$ contains the identity map $\Id : \cX \to \cX$.

When a set of transformations $\Phi \subseteq \cX^\cX$ is low-dimensional, we will assume that we are given the function $\vm$. We do {\em not} assume that we are given the set of matrices $\mK$ that correspond to functions in $\Phi$; rather, we will interpret $\vm$ as implicitly declaring the set $\Phi$ to contain {\em all} endomorphisms $\phi : \cX \to \cX$ of the form $\vx \mapsto \mK \vm(\vx)$. We will denote this set $\Phi_{\vm}$. 

\subsection{Concave games}

In a {\em concave game} $\Gamma$, each of $n$ players $i \in [n]$ has a convex compact \textit{strategy set} $\cX_i \subset \R^{d_i}$. A tuple $\vx \in \cX := \cX_1 \times \dots \times \cX_n \subset \R^d$, where $d = \sum_i d_i$ is called a {\em profile}. For each player $i$, a function $u_i : \cX \to [0, 1]$ defines the {\em utility} of player $i$. We assume that $u_i$ is concave in her own strategy $\vx_i \in \cX_i$ when the other players' strategies are held fixed.

We will sometimes be interested in problems in which the utilities have additional structure, namely, they are quadratic or linear. We will call a game {\em quadratic} if $u_i$ is a quadratic concave function in $\vx_i$ for every fixed opponent profile $\vx_{-i}$, \ie, if
$$u_i(\vx_i, \vx_{-i}) = \vb_i(\vx_{-i})^\top \vx_i - \vx_i^\top \mA_i(\vx_{-i}) \vx_i,$$
where $\vb_i : \cX_{-i} \to \R^{d_i}$ and $\mA_i : \cX_{-i} \to \S^{d_i}_+$ are arbitrary functions. We will call a game {\em linear} if furthermore $\mA_i \equiv \vec 0$, leaving only the linear term. 

A \emph{multilinear polytope game} is a special case in which each $\cX_i$ is a compact polytope and each utility is a multilinear polynomial, so that every monomial contains at most one coordinate of each player's strategy (raised to the first power). Thus, each utility is affine in a player's strategy when the other players' strategies are fixed. We represent each polytope by an explicit list of linear inequalities with rational coefficients and each utility by an explicit list of monomials with rational coefficients. We write $N$ for the total encoding length of this game description.

The main solution concepts of interest to us will be various notions of $\Phi$-equilibrium, which we now define. For each player $i$, let $\Phi_i \subseteq \cX_i^{\cX_i}$ be a set of transformations mapping the player's strategy set to itself. We will write $\Phi = (\Phi_1, \dots, \Phi_n)$. An {\em $\eps$-approximate $\Phi$-equilibrium} is a distribution $\mu \in \Delta(\cX)$ from which no player can profitably deviate by applying any function $\phi_i \in \Phi_i$. Formally:
\begin{align}
    \E_{\vx\sim\mu} \ab[ u_i(\phi_i(\vx_i), \vx_{-i}) - u_i(\vx) ] \le \eps
\end{align}
for every player $i$ and every deviation $\phi_i \in \Phi_i$. 

It is a standard result in no-regret learning and game theory that, if every player $i$ has $\Phi_i$-regret at most $\eps$, then the average strategy profile $\bar\mu := \frac{1}{T} \sum_{t=1}^T (\mu^{(t)}_1 \times \dots \mu^{(t)}_n)$ is an $\eps$-approximate $\Phi$-equilibrium. In particular, the notions of regret we singled out have corresponding notions of $\Phi$-equilibria: linear regret corresponds to {\em linear correlated equilibria}, and swap regret corresponds to {\em correlated equilibria}.

To ensure that expected utilities $\E_{\vx\sim\mu} u_i(\vx)$ and their gradients are efficiently computable (which is a basis for many of our algorithms), we will consider only games where the number of players $n$ is an absolute constant; otherwise, even the expected utility of a product distribution $\mu_1 \times \dots \mu_n$ may be hard to compute to high precision. Thus, throughout the paper, $\poly$ ignores the dependence on the number of players $n$. We elaborate on this issue in \Cref{sec:qefp,sec:low dim fptas} where it is relevant. Our lower bounds are not sensitive to this choice: all our lower bounds for equilibrium computation problems hold already for a constant number of players.

\subsection{Fixed point problems} We will reduce at various times from problems related to finding a fixed point of a contraction. Here, we fully specify these problems. Consider a convex compact set $\cX$ and a function $f : \cX \to \cX$ that is a $(1-\gamma)$-contraction under a norm $\norm{\cdot} : \R^{d_i} \to \R_{\ge 0}$, that is, $\norm{f(\vx) - f(\vy)} \le (1-\gamma) \norm{\vx - \vy}$ for all $\vx, \vy \in \cX$, where $\gamma > 0$. Banach's fixed point theorem guarantees, under these assumptions, that $f$ admits a unique fixed point $\vx^*$. We will always assume that $\max_{\vx, \vy \in \cX}\norm{\vx - \vy} \le 1$ for normalization; since we have assumed $B(\vec 0, 1) \subseteq \cX$, this also implies that $\norm{\cdot}$ is $1$-Lipschitz with respect to $\norm{\cdot}_2$.

\begin{definition}[\Contr]
Given $\cX, f, \norm{\cdot}, \gamma$, and a precision parameter $\delta > 0$, and promised that $f$ is a $(1-\gamma)$-contraction under norm $\norm{\cdot}$, output a point $\vx \in \cX$ for which $\norm{\vx - f(\vx)} \le \delta$. 
\end{definition}

Note that $\gamma$ and $\delta$ are given in binary, so they can be exponentially close to  $0$. We will also sometimes be interested in a generalization of \Contr, where the norm is {\em unknown} and in fact may not be a norm at all, merely an arbitrary convex function. We now define this problem. For a $1$-Lipschitz convex function $Q : \cX \to [0,1]$ and parameter $\gamma > 0$, we say that $f : \cX \to \cX$ is a $(1-\gamma)$-contraction under $Q$ if $Q(f(\vx)) \le (1-\gamma) Q(\vx)$ for every $\vx \in \cX$. Note that this immediately implies that $\min_{\vx \in \cX} Q(\vx) = 0$, by iterating $f$ and using compactness of $\cX$.

We call a problem {\em \Contr-hard} if it admits an efficient Turing reduction from \Contr.

It is well-known that Condon’s \emph{simple
stochastic games (SSGs)}~\citep{Condon92:Complexity} can be reduced to the problem of computing an $\epsilon$-fixed point of
a $(1 - \gamma)$-contraction over the hypercube under the $\ell_\infty$-norm; thus, $\SSG$ reduces to $\Contr$. Such a reduction was also shown by~\citet*{Etessami07:Complexity} for Shapley's stochastic games~\citep{Shapley53:Stochastic}, which is a broader class of games and has been instrumental in the development of multiagent reinforcement learning~\citep{Littman94:Markov}. In stark contrast, computing a fixed point of an $\ell_2$ contraction admits a $\polylog(1/\eps)$ algorithm via the ellipsoid algorithm~\citep{Huang99:Approximating,Sikorski93:Ellipsoid}. %
Indeed, this separation between $\ell_2$-contraction and other norm contractions is precisely what drives the separation between \Cref{cor:contr-polydim} (\Contr-hardness of $\Phi$-equilibria in concave games) and \Cref{theorem:quadratic-polytime} (efficient computation when utilities are quadratic): intriguingly, for $\Phi$-equilibria, the {\em utility function} in a game plays the role of the {\em norm} in the contraction problem!

\subsection{The PCP-for-PPAD conjecture}

We now state the precise version of the \emph{PCP-for-PPAD conjecture} of~\citet*{Babichenko16:Can} we need for~\Cref{th:multilinear hardness}. A \emph{generalized circuit} assigns a value in $[0,1]$ to each of $d$ nodes, each of which is the output of exactly one gate. Gates have at most two inputs and cycles are allowed. The circuit has \emph{fan-out} two if each node occurs in at most two input positions. All gate parameters $\beta$ are rational numbers in $[0,1]$. For input values $a,b$ and output value $c$, the arithmetic gates are satisfied at tolerance $\tau$ when the following conditions hold depending on the type of the gate:
\begin{itemize}
    \item \emph{constant}: $|c-\beta|\le\tau$;
    \item \emph{scaling}: $|c-\beta a|\le\tau$;
    \item \emph{copying}: $|c-a|\le\tau$;
    \item \emph{addition}: $|c-\min\{a+b,1\}|\le\tau$; and
    \item \emph{subtraction}: $|c-\max\{a-b,0\}|\le\tau$.
\end{itemize}
Finally, the comparison gate is satisfied when both implications
\begin{align}
    a<b-\tau\ \Longrightarrow\ c\ge1-\tau
    \text{ and } a>b+\tau\ \Longrightarrow\ c\le\tau
\end{align}
hold. We are ready to state the formal version of the conjecture.

\begin{assumption}[PCP-for-PPAD]\label{ass:pcp-ppad}
    There are absolute constants $0<\tau<1/4$ and $0<\gamma<1$ such that the following problem is \PPAD-hard: given a generalized circuit with fan-out two, compute an assignment satisfying at least $(1-\gamma)d$ gates at tolerance $\tau$. %
\end{assumption}

That is, this postulates that \PPAD-hardness of generalized circuits persists even when an arbitrary set of at most $\gamma d$ gates can fail.

\section{Fixed sets and convex expected fixed points}
\label{sec:fixed-sets}

In this section, we introduce the notion of a \emph{fixed set} and {\em convex expected fixed point}. These notions will be fundamental to several of our results. 

Let $\cX \subset \R^d$ be convex and compact, and $f : \cX \to \cX$ be a continuous function.

\begin{definition} \label{def:cefp}
An $\eps$-{\em convex expected fixed point} (hereafter {\em convex EFP}) of $f$ is a distribution $\mu \in \Delta(\cX)$ such that
\begin{align}
    \E_{\vx\sim\mu} \ab[Q(\vx) - Q(f(\vx))] \le \eps \label{eq:cefp-primal} \tag{CEFP}
\end{align}
for every convex $Q : \cX \to \R$ that is $1$-Lipschitz under $\norm{\cdot}_2$.
\end{definition}

Restricting $Q$ to linear $1$-Lipschitz functions gives the condition $\norm{\E_{\vx\sim\mu}[\vx-f(\vx)]}_2\le\eps$, known simply as an \emph{expected fixed point}~\citep{Zhang24:Efficient}, which was shown to admit a polynomial-time algorithm~\citep{Zhang25:Learning}.

\begin{definition}\label{def:fixed set}
    An $\eps$-{\em fixed set} of $f$ is a finite set $S = \{ \vx_1, \dots, \vx_k \} \subset \cX$ such that there is some Markov chain transition function $P : [k] \to \Delta([k])$, and stationary distribution $\mu$ of $P$, such that
\begin{align}
    \E_{i \sim \mu} \norm*{\vx_i - \E_{j \sim P(i)} f(\vx_j)}_2 \le \eps.
\end{align}
\end{definition}

An exact fixed set satisfies the condition $S \subseteq \conv f(S)$; the above notion turns out to be the appropriate notion of approximation.

Using duality, it turns out that the previous two definitions are equivalent.

\begin{theorem}\label{th:fixed set equivalence}
    The above two definitions are equivalent. That is, the support of an $\eps$-convex EFP is also an $\eps$-fixed set; conversely, given an $\eps$-fixed set, the corresponding stationary distribution $\mu$ is an $\eps$-convex EFP.
\end{theorem}

\begin{proof}
    Let $S$ be an $\eps$-fixed set, $P$ its corresponding transition matrix, and $\mu$ be its stationary distribution. Then for any convex $1$-Lipschitz $Q$, we have
    \begin{align}
        \E_{\vx} \ab[Q(\vx) - Q(f(\vx))] &= \E_\vx Q(\vx) - \E_\vx \E_{\vx'\mid\vx} Q(f(\vx'))
    \le \E_\vx \ab[Q(\vx) - Q\ab\Big(\E_{\vx'\mid\vx}  f(\vx'))]
        \le \E_\vx \norm[\Big]{\vx - \E_{\vx'\mid\vx} f(\vx')}_2 \le \eps,
    \end{align}
    where $\vx\sim\mu$ and $\vx'\mid\vx$ means $\vx'\sim P_\vx$. The first equality uses stationarity of $\mu$ under $P$. The inequalities are, in order, Jensen and Lipschitzness of $Q$.

    Conversely, let $\mu$ be an $\eps$-convex EFP with support $\{\vx_1, \dots, \vx_k\}$. We need to show that there is a transition matrix $P$ satisfying the fixed set condition. To do this, consider the following convex program.
    \begin{align}
        \min_{\mW \ge \vec 0} \quad \sum_i \norm*{\sum_j W_{ij} (\vx_i - f(\vx_j))}_2 \qq{s.t.} \mW^\top \vec 1 = \mW \vec 1 = \mu.
    \end{align}
    It suffices to show that the value of this program is $\le \eps$, because then it would be enough to use the transition matrix $P(j\mid i) = W_{ij}/\mu_i$.
    Using $\norm{\vv}_2=\max_{\norm{\va}_2\le1}\ip{\va,\vv}/d$, the dual of this problem is
    \begin{align}
        \max \E_{i \sim \mu} [b_i - c_i] \qq{s.t.} \frac1d\ip*{\va_i,\vx_i - f(\vx_j)} \ge b_i - c_j, \quad \norm{\va_i}_2 \le 1.
    \end{align}
    For any dual-feasible solution, we consider the function $$Q(\vx) = \max_i \{ b_i + \tfrac1d\ip{\va_i, \vx - \vx_i} \}.$$
    Then $Q$ is $1$-Lipschitz and convex. Moreover, we have $Q(\vx_i) \ge b_i$ and $Q(f(\vx_i)) \le c_i$,
    so
    \begin{align}
        \E_{i \sim \mu} [b_i - c_i] \le \E_{i \sim \mu} \ab[Q(\vx_i) - Q(f(\vx_i))] \le \eps. \tag*\qedhere
    \end{align}
\end{proof}

\begin{remark}
    An alternative way to see this equivalence is through Strassen's theorem~\citep{Gozlan17:Kantorovich}, which asserts that for probability measures $\mu,\nu$ on $\R^d$ with finite first moments, $\int Q\,d\mu\le\int Q\,d\nu$ for every convex and Lipschitz $Q$ if and only if they admit a coupling $(X,Y)$ with $\E[Y\mid X]=X$. The approximate version~\citep[Remark 3.4]{Gozlan17:Kantorovich} states that
    \begin{align}
        \int Q\,d\mu-\int Q\,d\nu\le\eps
        \quad\text{for every convex $1$-Lipschitz $Q$}
    \end{align}
    if and only if there is a coupling $X\sim\mu$, $Y\sim\nu$ satisfying
    \begin{align}
        \E\norm{X-\E[Y \mid X]}_2\le\eps.
    \end{align}
    To apply this to \Cref{th:fixed set equivalence}, let $\nu$ be the law of $f(X)$ for $X\sim\mu$. The convex EFP condition gives the coupling above. Lift $Y$ to $X'\sim\mu$ by sampling from $\mu$ conditional on $f(X')=Y$, independently of $X$ given $Y$. Then $f(X')=Y$, and $P(\vx'\mid\vx)=\Pr(X'=\vx'\mid X=\vx)$ has stationary distribution $\mu$, since both $X$ and $X'$ have law $\mu$. The coupling inequality is precisely the $\eps$-fixed-set condition. The converse follows directly from Jensen's inequality and Lipschitzness.
\end{remark}

\begin{remark}[Contractions]
    If $f$ is $L$-Lipschitz for $L<1$, with unique fixed point $\vx_*$, then every $\eps$-fixed set satisfies
    \begin{align}
        \E_{i \sim \mu} \norm{\vx_i - \vx_*}_2 \le \frac{\eps}{1-L}.
    \end{align}
    Indeed, applying the convex EFP condition to $Q(\vx) = \norm{\vx - \vx_*}_2$ gives
    \begin{align}
        (1-L)\E_{i \sim \mu} \norm{\vx_i - \vx_*}_2
        \le \E_{i \sim \mu}\ab[Q(\vx_i) - Q(f(\vx_i))] \le \eps.
    \end{align}
    In particular, for an exact fixed set, the stationary distribution $\mu$ is concentrated on $\vx_*$.
\end{remark}

\subsection{Correlated equilibria and fixed sets}

\begin{proposition}[Reducing fixed sets to correlated equilibria]
\label{prop:ce-fixed-set reduction}
    Given a map $f:[0,1]^d\to[0,1]^d$, we construct the two-player concave game on $\cX_1=\cX_2=[0,1]^d$ with utilities
    \begin{align}
        u_1(\vx,\vy)=-\norm{\vx-\vy}_2^2, \quad
        u_2(\vx,\vy)=-\norm{\vy-f(\vx)}_2^2.
    \end{align}
    For every $\eps\ge0$ and every $\eps$-CE $\mu$ of this game, $S=\supp\mu_\vx$ is a $2\sqrt{\eps}$-fixed set of $f$ with $|S|\le|\supp\mu|$, with a transition matrix whose stationary distribution is $\mu_\vx$.
\end{proposition}

\begin{proof}
    Each utility is concave and quadratic in its player's own strategy. The CE conditions imply
    \begin{align}\label{eq:ce-condition}
        \E_\vx \norm*{\vx - \E_{\vy\mid\vx} \vy}_2^2 \le \eps, \quad \E_\vy \norm*{\vy - \E_{\vx\mid\vy} f(\vx)}_2^2 \le \eps,
    \end{align}
    where $\E_\vx$ and $\E_\vy$ denote sampling from the marginals of $\mu$, and $\E_{\vx\mid\vy}$ and $\E_{\vy\mid\vx}$ denote sampling from the conditional distributions of $\mu$ given $\vy$ and $\vx$, respectively. Then
    \begin{align}
        \E_\vx \norm[\Big]{\vx - \E_{\vy\mid\vx, \vx'\mid\vy} f(\vx')}_2
        &\le \E_{\vx} \ab[\norm[\Big]{\vx - \E_{\vy\mid\vx} \vy}_2 + \norm[\Big]{\E_{\vy\mid\vx} \ab\big[\vy - \E_{\vx'\mid\vy} f(\vx')]}_2 ]
        \\&\le \E_\vx \ab[\norm[\Big]{\vx - \E_{\vy\mid\vx} \vy}_2 + \E_{\vy\mid\vx} \norm[\Big]{ \vy - \E_{\vx'\mid\vy} f(\vx')}_2]
        \\&\le \ab(\sqrt{\E_\vx \ab[\norm[\Big]{\vx - \E_{\vy\mid\vx} \vy}_2^2]} + \sqrt{\E_{\vy} \ab[\norm[\Big]{ \vy - \E_{\vx'\mid\vy} f(\vx')}_2^2]})
        \\&\le 2 \sqrt{\eps},
    \end{align}
    where we use, in order, the triangle inequality, Jensen, Cauchy-Schwarz, and finally \eqref{eq:ce-condition}.
    Set $S=\supp\mu_\vx$ and $P(\vx'\mid\vx)=\sum_\vy\mu(\vy\mid\vx)\mu(\vx'\mid\vy)$. This transition matrix has stationary distribution $\mu_\vx$, since
    \begin{align}
        \sum_\vx\mu_\vx(\vx)P(\vx'\mid\vx)
        =\sum_\vy\mu_\vy(\vy)\mu(\vx'\mid\vy)
        =\mu_\vx(\vx').
    \end{align}
    This is precisely the $2\sqrt{\eps}$-fixed-set condition.
\end{proof}

\subsection{Hardness of computing fixed sets}
\label{sec:hardness-fixedsets}

Computing a fixed set is certainly at most as hard as computing a fixed {\em point}, since fixed points are fixed sets. We now show that this equivalence goes the other way as well.

Let $f : [0, 1]^d \to [0, 1]^d$ be $L$-Lipschitz. It is known that computing a $\delta$-fixed point of $f$ is \PPAD-hard, even when $L\ge1$ and $0<\delta\le1$ are absolute constants~\citep{Rubinstein16:Settling}. Our goal is to prove the following.

\fixedsethardness*

The proof essentially makes the argument described in our technical overview quantitative. The remainder of this section is dedicated to its proof. To begin with, we construct the map 
\[
    g(\vx^{(1)}, \dots, \vx^{(r)}) = \ab(f_{\alpha_1}(\vx^{(1)}), \dots, f_{\alpha_r}(\vx^{(r)})),
\]
where the parameters $\alpha_1, \dots, \alpha_r$ will be chosen later. We let $S = \{ \vx_1, \dots, \vx_k \}$ be an $\eps$-fixed set of $g$ with $k\le K$, for a support bound $K \geq 2$ and precision $\eps > 0$ to be chosen later, and let $P$ and $\mu$ be the corresponding transition matrix and stationary distribution. In lieu of passing to the limit, we will select a layer $\ell \in [r]$ at which $\vx^{(\ell)}$ will have the same function as the limit $\vx$ in the proof of \Cref{prop:non-quantitative}. Our goal is to show that there is an $i \in [k]$ and $\ell \in [r]$ such that $\vx_i^{(\ell)}$ is a $\delta$-fixed point of $f$.

For a parameter $0<t\le1$, we define
\begin{align}
    P(t) \defeq \E_{a \sim \Geo(t)} P^a = t \sum_{a=0}^\infty (1-t)^a P^a = t(\mI-(1-t)P)^{-1}, \quad h(t) \defeq \E_{i \sim \mu} H(P_i(t))
\end{align}
where $H$ is entropy with natural logarithms and $P(1)=\mI$.
\begin{lemma}
    $h$ is nonincreasing.
\end{lemma}
\begin{proof}
    For $0<t<s\le1$, the stochastic matrix
    \begin{align}
        \mK(t,s)=\frac ts \mI+\ab(1-\frac ts)P P(t)
    \end{align}
    satisfies $P(t)=\mK(t,s)P(s)=P(s)\mK(t,s)$ and has stationary distribution $\mu$. But then we have
    \begin{align}
        h(t) - h(s) &= \E_{i \sim \mu} H(P_i(t)) - \E_{j \sim \mu}H(P_j(s))
        \\&= \E_{i \sim \mu, j \sim \mK_i} \ab[H(P_i(t)) - H(P_j(s))] \notag
        \\&= \E_{i \sim \mu, j \sim \mK_i} \KL(P_j(s) \lVert P_i(t)) \ge 0, \label{eq:entropy}
    \end{align}
    where $\KL(\cdot \lVert \cdot)$ denotes the KL divergence.
\end{proof}
Moreover, $P_i(t)$ is a distribution on $k$ items, so $h(t)$ is bounded in $[0, \log k]$. Thus, if we pick values $1 \eqdef t_0 > t_1 > t_2 > \dots > t_r > 0$, there will be some value $\ell \le r$ at which $$h(t_{\ell}) - h(t_{\ell-1}) \le \frac{\log k}{r}.$$
The use of entropy here is crucial, since we will later need the fact that the above bound is logarithmic in $k$. We fix this value of $\ell$, and for the rest of the proof, let $\vx_i \defeq \vx^{(\ell)}_i, \alpha \defeq \alpha_\ell, t \defeq t_\ell$.

\paragraph{Controlling the conditional variance.} For any $s \in [t, 1]$, we define
\begin{align}
    \bar{\vx}_i(s) \defeq \E_{j \sim P_i(s)} \vx_j, \quad V(s) \defeq \E_{i \sim \mu, j \sim P_i(s)} \norm{\vx_j - \bar{\vx}_i(s)}_2^2.
\end{align}
The quantity $V(s)$ is the variance of $\vx_j$ under $j\sim P_i(s)$, averaged over $i\sim\mu$; our goal is to show that $V(t)$ is small, so that these conditional distributions have small variance on average. We write
\begin{align}
    \vx_j = \E_{a \sim P_j}[ \vx_a + \alpha \Delta_a] + \ve_j,
\end{align}
where $\Delta_j = f(\vx_j) - \vx_j$. For a vector $\vv\in\R^{dr}$ with $d$-dimensional blocks, the normalized norms satisfy $\norm{\vv^{(\ell)}}_2\le\sqrt r\,\norm{\vv}_2$, so the mean error in this layer is at most $\eta \defeq \sqrt r\,\eps$ if $\norm{\vv} \leq \epsilon$. Thus,
\[
  \E_{j \sim \mu} \norm{\ve_j}_2 \le \eta.  
\]
For any $i$, taking expectations over $P_i(s)$ gives
\begin{align}\label{eq:inverse-identity}
    \E_{j \sim P_i(s),a \sim P_j}[ \alpha \Delta_a + \ve_j] &= \E_{j \sim P_i(s), a \sim P_j} [\vx_j - \vx_a] = \ip{P_i(s)(\mI-P), \vx},
\end{align}
where we use the notation $\ip{\mu, \vx} \defeq \sum_{i=1}^k \mu_i \vx_i$.

For $s<1$, recall that $P(s) = s(\mI - (1-s)P)^{-1}$.
Rearranging, we have
\begin{align}\label{eq:Ps-identity}
    P(s)(\mI-P) = \frac{s}{1-s}(\mI - P(s))
\end{align}
Substituting into \eqref{eq:inverse-identity} gives
\begin{align}\label{eq:error}
    \E_{j \sim P_i(s),a \sim P_j}[ \alpha \Delta_a + \ve_j] = \frac{s}{1-s}(\vx_i - \ip{P_i(s), \vx}) = \frac{s}{1-s}(\vx_i - \bar{\vx}_i(s)).
\end{align}
By triangle inequality, and taking expectations over $\mu$, it follows that
\begin{align}
    \E_{i \sim \mu} \norm{\vx_i - \bar{\vx}_i(s)}_2 \le \frac{1-s}{s} \E_{i \sim \mu, j \sim P_i(s),a \sim P_j} \ab[ \alpha \norm{\Delta_a}_2 + \norm{\ve_j}_2] \le \frac{1}{s}(\alpha + \eta) \label{eq:expected error bound}
\end{align}
where in the last inequality we use the trivial bounds $\norm{\Delta_a}_2 \le 1$ and $s \le 1$, as well as the fact that $\mu$ is stationary under $P(s)$.

Since $\bar{\vx}_i(s)$ is the mean of $\vx_j$ over $P_i(s)$, it minimizes the mean-squared error in $V(s)$. Therefore, we have
\begin{align}
    V(s) &\le \E_{i \sim \mu, j \sim P_i(s)} \norm{\vx_j - \vx_i}_2^2
    \\&= \frac{2}{d}\E_{i \sim \mu, j \sim P_i(s)} \ip{\vx_i - \vx_j, \vx_i}
    \\&= \frac{2}{d}\E_{i \sim \mu} \ip{\vx_i - \bar{\vx}_i(s), \vx_i}
    \\&\le 2 \E_{i \sim \mu} \norm{\vx_i - \bar{\vx}_i(s)}_2
    \\&\le \frac{2}{s}(\alpha + \eta).
\end{align}
where the first equality uses the fact that $i$ and $j$ have the same distribution (because $\mu$ is stationary), the second uses the definition of $\bar{\vx}_i(s)$, and the final two inequalities use Cauchy-Schwarz with $\norm{\vx_i}_2 \le 1$ and \eqref{eq:expected error bound}.

Now set $\mK \defeq \mK(t, s)$. Then
\begin{align}\label{eq:K transfer}
    \bar{\vx}_i(t) = \E_{j \sim P_i(t)} \vx_j = \E_{a \sim \mK_i, j \sim P_a(s)} \vx_j = \E_{j \sim \mK_i} \bar{\vx}_j(s).
\end{align}
Therefore,
\begin{align}
    V(t)-V(s) &= \E_{i\sim\mu}\ab[\norm{\bar{\vx}_i(s)}_2^2-\norm{\bar{\vx}_i(t)}_2^2]
    \\&= \E_{i \sim \mu, j \sim \mK_i} \norm{\bar{\vx}_j(s) - \bar{\vx}_i(t)}_2^2
    \\&= \E_{i \sim \mu, j \sim \mK_i} \norm*{\sum_a (P_{ja}(s) - P_{ia}(t)) \vx_a}_2^2
    \\&\le \E_{i \sim \mu, j \sim \mK_i} \ab[\ab(\sum_a \abs{P_{ja}(s) - P_{ia}(t)} \norm*{\vx_a}_2)^2]
    \\&\le \E_{i \sim \mu, j \sim \mK_i} \ab(\sum_a |P_{ja}(s)-P_{ia}(t)|)^2
    \\&\le 2\E_{i \sim \mu, j \sim \mK_i} \KL(P_j(s) \lVert P_i(t))
    \\&= 2(h(t)-h(s)).
\end{align}
The first equality uses stationarity and the variance identity; the second uses \eqref{eq:K transfer} and stationarity under $\mK$. The inequalities use the triangle inequality, $\norm{\vx_a}_2\le1$, and Pinsker's inequality for the unnormalized sum of absolute probability differences. The final equality is \eqref{eq:entropy}. For $s=1$, $V(1)=0$, so the preceding bound on $V(s)$ also holds at this endpoint.

Taking $t = t_\ell$ and $s = t_{\ell-1}$ therefore gives
\begin{align}
    V(t) \le \frac{2\log k}{r} + V(s) \le \frac{2\log k}{r} + \frac{2}{s}(\alpha + \eta).
\end{align}

\paragraph{Dividing by $\alpha$.} Starting from \eqref{eq:error}, taking $s=t$, rearranging, and dividing by $\alpha$ yields
\begin{align}
    \E_{j \sim P_i(t),a \sim P_j} \Delta_a =\frac{t}{\alpha(1-t)}(\vx_i - \bar{\vx}_i(t)) - \frac{1}{\alpha}\E_{j \sim P_i(t),a \sim P_j} \ve_j.
\end{align}
On the other hand, applying \eqref{eq:Ps-identity} to ($\Delta_1, \dots, \Delta_k$) gives
\begin{align}
    \E_{j \sim P_i(t)} \Delta_j - \E_{j \sim P_i(t), a \sim P_j} \Delta_a = \frac{t}{1-t} \ab(\Delta_i - \E_{j \sim P_i(t)} \Delta_j).
\end{align}
Adding and using the triangle inequality with $t\le1/2$, as ensured by our parameter choices below, and the bounds $\alpha\le1$ and $\norm{\vx}_2\le1$ for $\vx\in[0,1]^d$ gives
\begin{align}
    \norm*{\E_{j \sim P_i(t)} \Delta_j}_2 &\le \frac{2t}{\alpha}\norm{\vx_i - \bar{\vx}_i(t)}_2 + \frac{1}{\alpha}\E_{j \sim P_i(t),a \sim P_j}\norm{\ve_j}_2 + 2t \ab(\norm{\Delta_i}_2 + \E_{j \sim P_i(t)} \norm{\Delta_j}_2).
    \\&\le \frac{8t}{\alpha} + \frac{1}{\alpha}\E_{j \sim P_i(t),a \sim P_j}\norm{\ve_j}_2.
\end{align}
Averaging over $i \sim \mu$, we have
\begin{align}
    \E_{i \sim \mu} \norm*{\E_{j \sim P_i(t)} \Delta_j}_2 \le \frac{8t + \eta}{\alpha}\eqdef \gamma
\end{align}
This is the quantitative version of \eqref{eq:qualitative-efp}.

\paragraph{Using continuity of $f$.}
The final step is to use continuity of $f$ to conclude that we have found a fixed point. For independent draws $j,a\sim P_i(t)$, we have
\begin{align}
    \frac12 \E_{j, a \sim P_i(t)}\norm{\vx_j - \vx_a}_2^2 = \E_{j \sim P_i(t)} \norm{\vx_j - \bar{\vx}_i(t)}_2^2.
\end{align}
Here, $\bar{\vx}_i(t)$ is the mean of $\vx_j$ over $P_i(t)$. Therefore, we have
\begin{align}
    \E_{j \sim P_i(t)} \norm{\Delta_j}_2 &\le \sqrt{\E_{j \sim P_i(t)} \norm{\Delta_j}_2^2}
    \\&\le \norm*{\E_{j \sim P_i(t)}\Delta_j}_2 + \sqrt{\frac12 \E_{j, a \sim P_i(t)} \norm{\Delta_j - \Delta_a}_2^2}
    \\&\le \norm*{\E_{j\sim P_i(t)}\Delta_j}_2 + 2L\sqrt{\frac12 \E_{j, a \sim P_i(t)} \norm{\vx_j - \vx_a}_2^2}
    \\&\le \norm*{\E_{j\sim P_i(t)}\Delta_j}_2 + 2L\sqrt{\E_{j \sim P_i(t)} \norm{\vx_j - \bar{\vx}_i(t)}_2^2}.
\end{align}
where we use, in order: Cauchy-Schwarz; the inequality $\sqrt{a + b} \le \sqrt{a} + \sqrt{b}$ together with variance decomposition; Lipschitzness of $f$ with the bound $L \ge 1$; and finally the previous identity.
Taking expectations over $\mu$, using the fact that $\mu$ is stationary, and applying Jensen now yields
\begin{align}\label{eq:amplification residual}
\E_{i \sim \mu} \norm{\Delta_i}_2 &\le \gamma + 2L\sqrt{V(t)} \le \frac{8t + \eta}{\alpha} + 2L\sqrt{\frac{2\log k}{r} + \frac{2(\alpha + \eta)}{s}}.
\end{align}
We now take $\alpha_\ell=b^{-\ell}$ and set $t_\ell=\delta\alpha_\ell/64$, for a base $b\ge2$ to be chosen below. This makes $8t/\alpha=\delta/8$. Since $s=t_{\ell-1}$ and $t_0=1$, we also have $\alpha/s\le64/(\delta b)$ and $1/s\le64b^{r-1}/\delta$. Substituting in \eqref{eq:amplification residual} gives
\begin{align}
    \E_{i\sim\mu}\norm{\Delta_i}_2
    &\le\frac\delta8+\eta b^r
      +2L\sqrt{\frac{2\log k}{r}+\frac{128}{\delta b}+\frac{128\eta b^{r-1}}\delta}.
\end{align}
The remaining parameter choices bound the three terms under the square root. Taking $r=(1+\ceil{\log_2 K})\ceil{128(L+1)^2/\delta^2}$ bounds the entropy term by $\delta^2/(64(L+1)^2)$, since $k\le K$. Taking $b=\ceil{\max\{2,2^{14}(L+1)^2/\delta^3\}}$ bounds the next term by $\delta^2/(128(L+1)^2)$. Finally, we choose $\eps=\delta/(16rb^r)$. Recalling that $\eta=\sqrt r\,\eps$, this gives $\eta b^r\le\delta/16$ and bounds the last term under the square root by $\delta^2/(2048(L+1)^2)$. Thus,
\begin{align}
    \E_{i\sim\mu}\norm{\Delta_i}_2
    &\le\frac{3\delta}{16}
      +2L\sqrt{\frac{\delta^2}{64(L+1)^2}+\frac{\delta^2}{128(L+1)^2}+\frac{\delta^2}{2048(L+1)^2}}
      <\frac\delta2.
\end{align}
This shows that a supported point in this layer has residual less than $\delta/2$. For fixed $L,\delta$, we have $r=O(\log K)$ and $b = \Theta(1)$, so $\epsilon = \delta/(16 r b^r) \geq K^{-c_0}$ for a sufficiently large absolute constant $c_0$. The construction takes time polynomial in the input and $\log K$, and extracting a fixed point takes time polynomial in the returned support encoding. To obtain the formulation stated in the theorem, for a fixed $C \geq 1$, we take $K = d^{2 \lceil C \rceil}$.  The constructed map $g$ has dimension $D=dr=O(d\log d)$, so $D^C\le K$ for all sufficiently large $d$. Taking $c=2\ceil{C}c_0$, we also have $D^{-c}\le d^{-c}=K^{-c_0}\le\eps$. This proves~\Cref{th:fixed set hardness}.

\begin{remark}[Beyond explicit support]
    \label{rem:sampling representations}
    Instead of explicitly listing each point in the support, one may instead take the output of the algorithm to be a probabilistic Turing machine, which samples a stationary distribution $\mu$ corresponding to an $\eps$-fixed set. For a support bound $K\ge 2$ and accuracy $\eps=K^{-c}$ for a sufficiently large absolute constant $c$, the proof of \Cref{th:fixed set hardness}, applied with accuracy $\delta/2$, uses $O(\log K)$ layers and, with suitable parameter choices, gives $\E_{\vx\sim\mu}\norm{f(\vx^{(\ell)})-\vx^{(\ell)}}_2\le\delta/4$ for some layer $\ell$. By Markov's inequality, two independent samples contain a point with residual at most $\delta/2$ in that layer with probability at least $3/4$. Thus, we can test every layer to obtain a $\delta$-fixed point with probability at least $3/4$, without access to the transition matrix. As a result, the same proof establishes hardness under randomized reductions even for superpolynomial support, provided $\log K$ is polynomial in the input size.
\end{remark}

\begin{remark}\label{rem:fixed-set-fptas}
    The hardness in \Cref{th:fixed set hardness} applies when the accuracy is small enough relative to a prescribed support bound. Without that restriction, approximate fixed sets admit a simple FPTAS: start from any point $\vx_1=\mathbf{0}$, and at each step, set $\vx_{t+1} = f(\vx_t)$. With $O(1/\epsilon)$ iterations, the uniform distribution over that sequence of points is an $\epsilon$-fixed set.
\end{remark}

\section{Hardness of correlated equilibria in multilinear games}
\label{sec:polytope}

We next consider games in which each strategy set $\cX_i$ is a compact rational polytope, given by an explicit list of linear inequalities, and each utility is an explicit multilinear polynomial. Our goal is to prove~\Cref{th:multilinear hardness}, which we restate below.

\polytopemult*

The basic idea is to implement quadratic imitation utilities using linear costs on a larger strategy polytope, through the use of auxiliary coordinates. We then use a constant number of additional players to evaluate a polynomial fixed-point map through a multilinear game, using an additional gadget to enforce agreement among the copies.

\subsection{Exponential precision}

\begin{proof}[Proof of \Cref{th:multilinear hardness}(i)] To begin with, we convert the hard fixed-point problem into one with a polynomial function. This step is standard, but we include it for completeness.

\paragraph{A polynomial map.} It is known that satisfying a generalized circuit with fan-out two is \PPAD-hard at some constant tolerance $0<\tau<1/4$~\citep{Rubinstein15:Inapproximability}. Let $d$ be its number of nodes and $n \ge d$ its encoding length. Each gate will be replaced by a continuous function of its inputs. For arithmetic gates, we use the clipping operation to make sure that the output stays in $[0,1]$. For Boolean and comparison gates, we write $a,b\in[0,1]$ for the input values, and define
\begin{align}
    R_\tau(x) = \clip_{[0,1]}\ab(\frac{x-\tau}{1-2\tau}).
\end{align}
This rescaling ensures that approximately Boolean inputs yield exact Boolean outputs, leaving some slack for the subsequent approximation. We will use the following gates.\footnote{\citet*{FilosRatsikas23:Complexity} have shown that $\PPAD$-hardness holds with only (truncated) addition and complement gates; that simplification is not important for our purposes here.}
\begin{itemize}
    \item Comparison: $h_{<}(a,b) = \clip_{[0,1]}\ab(\frac12+\frac{b-a}{2\tau})$.
    \item OR: $h_{\mathrm{OR}}(a,b) = R_\tau(\max\{a,b\})$.
    \item AND: $h_{\mathrm{AND}}(a,b) = R_\tau(\min\{a,b\})$.
    \item NOT: $h_{\mathrm{NOT}}(a) = 1-R_\tau(a)$.
\end{itemize}
For a constant $\lambda \ge1$, each gate map satisfies $|h_i(a,b)-h_i(a',b')|\le\lambda(|a-a'|+|b-b'|)$, and each variable occurs in at most three input positions. We will use a sufficiently accurate constant-degree polynomial approximation of each $h_i$, making sure that it preserves the range $[0,1]$ and has bounded derivatives. In particular, we make use of the Bernstein approximation~\citep{Lorentz12:Bernstein}:
\begin{align}
    f_i(a,b) = \E_{\xi_a \sim \operatorname{Bin}(p,a),\,\xi_b \sim \operatorname{Bin}(p,b)} h_i(\xi_a/p,\xi_b/p),
\end{align}
where the binomial variables above are independent. We will use the following basic properties of the Bernstein approximation~\citep{Lorentz12:Bernstein}; we include the short proofs for completeness.

\begin{claim}
    \label{cl:bernstein approximation}
    For every integer $p\ge1$, $f_i$ is a polynomial of degree at most $p$ in each input (for a total degree at most $2p$). For all $(a,b)\in[0,1]^2$,
    \begin{align}
        0\le f_i(a,b)\le1,\quad
        |f_i(a,b)-h_i(a,b)|\le\frac{\lambda}{\sqrt p},\quad
        |\partial_a f_i(a,b)|,\ |\partial_b f_i(a,b)|\le\lambda.
    \end{align}
\end{claim}

\begin{proof}
The range bound follows because $f_i$ averages values in $[0,1]$. The degree bound follows by expanding the binomial probabilities. Lipschitzness of $h_i$ together with the binomial variance bound give
\begin{align}
    |f_i(a,b)-h_i(a,b)|
    \le\lambda\ab(\E|\xi_a/p-a|+\E|\xi_b/p-b|)
    \le\lambda\ab(\sqrt{\frac{a(1-a)}p}+\sqrt{\frac{b(1-b)}p})
    \le\frac{\lambda}{\sqrt p}.
\end{align}
The derivative bounds follow by differentiation and the Lipschitznes of $h_i$.
\end{proof}

In what follows, we fix an integer $p\ge(4\lambda/\tau)^2$. By \Cref{cl:bernstein approximation}, the approximation error is at most $\tau/4$. Because every input appears at most three times, it follows that the Euclidean operator norm of $f$ is at most $\sqrt6\lambda\le3\lambda$, giving a Lipschitz bound $L=3\lambda$ for $f$ in the normalized norm as well.

We now set $\delta=\tau/(4d)$. If $\norm{f(\vx)-\vx}_2 \le \delta$, each coordinate residual is at most $\sqrt d\,\delta \le \tau/4$. Accounting for the approximation error, this satisfies every original gate at tolerance $\tau$. Thus, finding such a point is \PPAD-hard.

\paragraph{Amplifying the support bound.}
We next apply the amplification argument from \Cref{sec:hardness-fixedsets}. In particular, if $K \geq 2$ is the support, we set
\begin{align}\label{eq:multilinear amplification parameters}
    r \defeq (1+\ceil{\log_2 K})\ceil{128(L+1)^2/\delta^2} \text{ and } b \defeq \ceil{\max\{2,2\lambda+1,2^{14}(L+1)^2/\delta^3\}},
\end{align}
and define $g:[0,1]^{dr}\to[0,1]^{dr}$ by
\begin{align}
    g(\vx^{(1)},\dots,\vx^{(r)})
    = \ab((1-b^{-\ell})\vx^{(\ell)}+b^{-\ell}f(\vx^{(\ell)}))_{\ell=1}^r.
\end{align}
Moreover, $\zeta \defeq \delta/(16rb^r)$. The following bound follows directly from the argument in~\Cref{sec:hardness-fixedsets}.

\begin{claim}[Extracting a fixed point]\label{cl:multilinear amplification}
    If $\mu$ has support size at most $K$ and satisfies $\E[Q(\vx)-Q(g(\vx))]\le\zeta$ for every convex $1$-Lipschitz $Q$, then some $\vx\in\supp\mu$ and $\ell\in[r]$ satisfy $\norm{f(\vx^{(\ell)})-\vx^{(\ell)}}_2\le\delta$.
\end{claim}

\paragraph{Encoding the quadratic penalty into the constraints.} Let $T: [0, 1] \ni x \mapsto \min\{2x,2(1-x)\}$ and
\begin{align}
    p_m(x)=x-\sum_{j=1}^m4^{-j}T^{(j)}(x),
\end{align}
where $T^{(j)}$ denotes the $j$th composition of $T$ with itself. The function $p_m$ linearly interpolates $x^2$ on the mesh $2^{-m}\Z\cap[0,1]$, as illustrated in~\Cref{fig:quadratic-interpolation}. On a mesh interval $[a,b]$, its error is $(x-a)(b-x)$, so
\begin{align}\label{eq:multilinear square}
    x^2 \le p_m(x) \le x^2+4^{-(m+1)}.
\end{align}

\begin{figure}[htbp]
    \centering
    \begingroup
\definecolor{interpOrange}{HTML}{C65E32}
\definecolor{interpTeal}{HTML}{237B87}
\begin{tikzpicture}[
    font=\small,
    >=Stealth,
    axis/.style={->,black,line width=0.5pt},
    guide/.style={black!18,line width=0.4pt},
    quadratic/.style={black!75,densely dashed,line width=0.6pt},
    coarse/.style={interpOrange,line width=0.8pt},
    fine/.style={interpTeal,line width=0.8pt}
]
    \begin{scope}[x=5.8cm,y=4.1cm]
        \node at (0.5,1.19) {(a) Interpolation};
        \foreach \v in {0.25,0.5,0.75,1} {
            \draw[guide] (0,\v) -- (1,\v);
        }
        \fill[interpOrange!9]
            (0,0) -- (0.5,0.25) -- (1,1)
            -- plot[domain=1:0,samples=81,variable=\t] (\t,{\t*\t})
            -- cycle;
        \draw[axis] (0,0) -- (1.08,0) node[right] {$x$};
        \draw[axis] (0,0) -- (0,1.09);
        \foreach \v/\lab in {0/0,0.25/{\tfrac14},0.5/{\tfrac12},0.75/{\tfrac34},1/1} {
            \draw[black] (\v,0) -- ++(0,-0.018)
                node[below] {$\lab$};
        }
        \foreach \v/\lab in {0.25/{\tfrac14},0.5/{\tfrac12},0.75/{\tfrac34},1/1} {
            \draw[black] (0,\v) -- ++(-0.012,0)
                node[left] {$\lab$};
        }
        \draw[coarse] (0,0) -- (0.5,0.25) -- (1,1);
        \draw[fine] (0,0) -- (0.25,0.0625) -- (0.5,0.25)
            -- (0.75,0.5625) -- (1,1);
        \draw[quadratic] plot[domain=0:1,samples=81,variable=\t]
            (\t,{\t*\t});
        \foreach \v in {0,0.25,0.5,0.75,1} {
            \fill[interpTeal] (\v,{\v*\v}) circle[radius=1.6pt];
        }
    \end{scope}

    \begin{scope}[xshift=8cm,x=5.8cm,y=65.6cm]
        \node at (0.5,0.074375) {(b) Error $p_m(x)-x^2$};
        \foreach \v/\lab in {0.015625/{\tfrac1{64}},0.0625/{\tfrac1{16}}} {
            \draw[guide] (0,\v) -- (1,\v);
            \draw[black] (0,\v) -- ++(-0.012,0)
                node[left] {$\lab$};
        }
        \draw[axis] (0,0) -- (1.08,0) node[right] {$x$};
        \draw[axis] (0,0) -- (0,0.068125);
        \foreach \v/\lab in {0/0,0.25/{\tfrac14},0.5/{\tfrac12},0.75/{\tfrac34},1/1} {
            \draw[black] (\v,0) -- ++(0,-0.001125)
                node[below] {$\lab$};
        }
        \foreach \j in {0,1} {
            \pgfmathsetmacro{\a}{\j/2}
            \pgfmathsetmacro{\b}{(\j+1)/2}
            \draw[coarse] plot[domain=\a:\b,samples=41,variable=\t]
                (\t,{(\t-\a)*(\b-\t)});
        }
        \foreach \j in {0,1,2,3} {
            \pgfmathsetmacro{\a}{\j/4}
            \pgfmathsetmacro{\b}{(\j+1)/4}
            \fill[interpTeal!10] (\a,0)
                -- plot[domain=\a:\b,samples=31,variable=\t]
                    (\t,{(\t-\a)*(\b-\t)})
                -- (\b,0) -- cycle;
            \draw[fine] plot[domain=\a:\b,samples=31,variable=\t]
                (\t,{(\t-\a)*(\b-\t)});
        }
    \end{scope}

    \draw[quadratic] (1.1,-0.85) -- (1.8,-0.85)
        node[right,black] {$x^2$};
    \draw[coarse] (4.3,-0.85) -- (5,-0.85)
        node[right,black] {$m=1$: two pieces};
    \draw[fine] (9,-0.85) -- (9.7,-0.85)
        node[right,black] {$m=2$: four pieces};
\end{tikzpicture}
\endgroup
    \caption{The piecewise-linear interpolants $p_1$ and $p_2$ relative to $x^2$. Each successive refinement doubles the number of pieces and divides the maximum error by four.}
    \label{fig:quadratic-interpolation}
\end{figure}

Although this interpolant has exponentially many pieces, it will suffice to consider the following compact polytope:
\begin{align}\label{eq:multilinear lift}
    \begin{gathered}
        0\le x,t,w_1,\dots,w_m\le1,\quad w_0=x,\\
        w_j\le2w_{j-1},\quad w_j\le2(1-w_{j-1})\quad \text{for } j\in[m],\quad
        t\ge x-\sum_{j=1}^m4^{-j}w_j.
    \end{gathered}
\end{align}
Crucially, in the game constructed below, $t$ enters the player's utility as a linear cost, so the player will be incentivized to choose a lift minimizing $t$. This choice recovers $p_m(x)$, as formalized below.

\begin{claim}
    \label{cl:multilinear lift}
    For every fixed $x\in[0,1]$, the minimum feasible $t$ in \eqref{eq:multilinear lift} is $p_m(x)$, attained when $w_j=T^{(j)}(x)$.
\end{claim}

\begin{proof}
We show that minimizing $t$ forces $w_j=T(w_{j-1})$. For the sake of contradiction, we consider the first coordinate $w_j$ for which $w_j < T(w_{j-1})$. If $w_j$ is increased by $\Delta$, later coordinates can be reduced to retain feasibility, losing at most $2^{s-j}\Delta$ at position $s > j$. Indeed, after setting $w_j'=w_j+\Delta$, recursively choosing $w_s'=\min\{w_s,T(w_{s-1}')\}$ preserves feasibility, and the $2$-Lipschitz property of $T$ implies $0 \le w_s-w_s'\le2^{s-j}\Delta$ by induction. As a result, the gain in $\sum_s4^{-s}w_s$ is at least
\begin{align}
    4^{-j}\Delta\ab(1-\sum_{h=1}^{m-j}2^{-h})>0,
\end{align}
which is a contradiction. We conclude that every minimizing lift satisfies $w_j=T^{(j)}(x)$ and $t=p_m(x)$, as claimed.
\end{proof}

Let $\cS$ be the product of $dr$ copies of \eqref{eq:multilinear lift}. We write a lifted strategy as $\vs=(\vx,\vt,\vw_1,\dots,\vw_m)\in\cS$, where each of these vectors lies in $[0,1]^{dr}$. We denote by $\vx(\vs)=\vx$ its projection onto the original coordinates and $c(\vs) \defeq (dr)^{-1}\sum_a t_a$ its linear cost. The basic property we need is that
\begin{align}\label{eq:multilinear cost}
    c(\vs)\ge\norm{\vx(\vs)}_2^2,
    \qq{and every $\vx$ has a lift with}
    c(\vs)\le\norm{\vx}_2^2+4^{-(m+1)}.
\end{align}
We have therefore replaced the quadratic cost by a linear one using only $O(drm)$ variables and inequalities.

\paragraph{Evaluating the map with copies.} To implement each $f_i(a,b)$ as a multilinear function, we write each binomial variable as a sum of $p$ independent Bernoulli variables controlled by different players. In particular, if coordinate $i$ has inputs $a_i,b_i$, we set
\begin{align}
    \widetilde f_i(\vz^{(1)},\dots,\vz^{(2p)})
    =\E h_i\ab(\frac1p\sum_{j=1}^p \xi_j,\frac1p\sum_{j=p+1}^{2p}\xi_j),
\end{align}
where the $\xi_j$ are independent, with success probabilities $z^{(j)}_{a_i}$ for $j\le p$ and $z^{(j)}_{b_i}$ otherwise. This multilinear polynomial takes values in $[0,1]$ even when the copies disagree, and equals $f_i(\vx)$ when all copies equal $\vx$.

Moreover, we use one more copy for the retained coordinate in each layer of $g$, for a total of $q=2p+1$ copies. Writing $\vz^{(j,\ell)}$ for the $\ell$th $d$-dimensional block of copy $j$, we define
\begin{align}
    \widetilde g^{(\ell)}(\vz^{(1)},\dots,\vz^{(q)})
    =(1-b^{-\ell})\vz^{(q,\ell)}
    +b^{-\ell}\widetilde f(\vz^{(1,\ell)},\dots,\vz^{(2p,\ell)}).
\end{align}
This map takes values in $[0,1]^{dr}$ and equals $g(\vx)$ when all copies equal $\vx$. Since each output uses at most $q$ inputs and each input occurs at most three times in a fixed copy, Cauchy-Schwarz yields
\begin{align}\label{eq:multilinear copies}
    \norm{\widetilde g(\vz^{(1)},\dots,\vz^{(q)})-g(\vx)}_2
    \le 3q\sqrt{\sum_{j=1}^q\norm{\vz^{(j)}-\vx}_2^2}.
\end{align}

\paragraph{The final construction.} The game contains $q+2$ players, each of whom has strategy set $\cS$. For a strategy profile $(\vs_x,\vs_y,\vs_1,\dots,\vs_q)$, we write $\vx=\vx(\vs_x)$, $\vy=\vx(\vs_y)$, and $\vz^{(j)}=\vx(\vs_j)$ for the projected vectors. All players use the cost function $c$ defined above.

The roles of $\vx$ and $\vy$ are inherited from the quadratic game in \Cref{sec:hardness-fixedsets}: $\vx$ tries to imitate $\vy$, while $\vy$ tries to imitate the map evaluated at $\vx$. To make the utilities multilinear, we first replace the quadratic costs by the lifted linear cost $c$ and then evaluate the map using $\widetilde g(\vz^{(1)},\dots,\vz^{(q)})$, where each copy player is imitating $\vx$. Moreover, we reward $\vx$ for matching the copies, assigning a weight $M\ge1$ per copy. When $M \gg$ 1 and the imitation is sufficiently accurate, each copy will be forced to be close to $\vx$. At the same time, we retain the incentive for $\vx$ to imitate $\vy$.

In this context, we define the utilities
\begin{align}\label{eq:multilinear utilities}
    u_x(\vs_x,\vs_y,\vs_1,\dots,\vs_q)
    &=\frac{2}{dr(1+qM)}\ip{\vx,\vy+M\sum_j\vz^{(j)}}-c(\vs_x),\\
    u_y(\vs_x,\vs_y,\vs_1,\dots,\vs_q)
    &=\frac{2}{dr}\ip{\vy,\widetilde g(\vz^{(1)},\dots,\vz^{(q)})}-c(\vs_y),\\
    u_j(\vs_x,\vs_y,\vs_1,\dots,\vs_q)
    &=\frac{2}{dr}\ip{\vz^{(j)},\vx}-c(\vs_j)\quad \text{for } j\in[q].
\end{align}
(The factors $(dr)^{-1}$ are used per our normalized norm convention.) Each utility is multilinear and lies in $[-1,2]$. Shifting and scaling would place the utilities in $[0,1]$, changing the CE accuracy by only a constant factor, so we can work directly with the utilities above.

\begin{claim}
    \label{cl:multilinear extraction}
    Let $\mu$ be an $\eps$-correlated equilibrium of the constructed game and $\rho=\sqrt{\eps+4^{-(m+1)}}$. The $\vx$-marginal of $\mu$ is an $\eta$-convex EFP of $g$, where
    \begin{align}\label{eq:multilinear extraction}
        \eta=2(1+qM)\rho+3q\sqrt{M^{-1}+(2q+1)\rho}.
    \end{align}
\end{claim}

\begin{proof}
We define the deviations
\begin{align}
    \widehat{\vx}=\E\ab[\frac{\vy+M\sum_j\vz^{(j)}}{1+qM}\,\middle|\,\vs_x], \quad
    \widehat{\vz}^{(j)}=\E[\vx\mid\vs_j], \quad
    \widehat{\vy}=\E[\widetilde g(\vz^{(1)},\dots,\vz^{(q)})\mid\vs_y].
\end{align}
For each recommendation $\vs_x$, we let $\vs'_x$ be the minimum-cost lift of $\widehat{\vx}$. \Cref{cl:multilinear lift} gives
\begin{align}
    \vx(\vs'_x)=\widehat{\vx},\quad c(\vs'_x)\le\norm{\widehat{\vx}}_2^2+4^{-(m+1)}.
\end{align}
Writing $\vs_{-x}$ for the other players' recommendations, the conditional utility gain is
\begin{align}
    \E\ab[u_x(\vs'_x,\vs_{-x})-u_x(\vs_x,\vs_{-x})\,\middle|\,\vs_x] &=\frac{2}{dr}\ip{\widehat{\vx}-\vx,\widehat{\vx}}+c(\vs_x)-c(\vs'_x)\\
    &\ge 2\norm{\widehat{\vx}}_2^2-\frac{2}{dr}\ip{\vx,\widehat{\vx}}+c(\vs_x)-\norm{\widehat{\vx}}_2^2-4^{-(m+1)}\\
    &=\norm{\vx-\widehat{\vx}}_2^2+c(\vs_x)-\norm{\vx}_2^2-4^{-(m+1)}\\
    &\ge\norm{\vx-\widehat{\vx}}_2^2-4^{-(m+1)},
\end{align}
where the first equality uses the definition of $\widehat{\vx}$ and the last inequality uses $c(\vs_x)\ge\norm{\vx}_2^2$. Taking expectations and using the fact that $\mu$ is an $\epsilon$-correlated equilibrium,
\begin{align}
    \E\norm{\vx-\widehat{\vx}}_2^2 \leq 4^{-(m+1)}  +\eps = \rho^2.
\end{align}
Similarly, we obtain
\begin{align}\label{eq:multilinear obedience}
    \E\norm{\vx-\widehat{\vx}}_2^2\le\rho^2,\quad
    \E\norm{\vz^{(j)}-\widehat{\vz}^{(j)}}_2^2\le\rho^2,\quad
    \E\norm{\vy-\widehat{\vy}}_2^2\le\rho^2.
\end{align}
We now argue that the copies are close to $\vx$ in mean square. Indeed, we have
\begin{align}
    (1+qM)\E\ip{\vx,\widehat{\vx}} =\E\ip{\vx,\vy}+M\sum_j\E\ip{\vx,\vz^{(j)}}, \quad \E\ip{\vz^{(j)},\widehat{\vz}^{(j)}} = \E\ip{\vz^{(j)},\vx}.
\end{align}
Thus,
\begin{align}
    M\sum_j\E\norm{\vz^{(j)}-\vx}_2^2
    &= \frac1{dr}\E\ip{\vx,\vy-\vx}
    +\frac{1+qM}{dr}\E\ip{\vx,\vx-\widehat{\vx}} + \frac{M}{dr}\sum_j\E\ip{\vz^{(j)},\vz^{(j)}-\widehat{\vz}^{(j)}} \\
    &\le 1+(1+2qM)\rho.
\end{align}
To establish the inequality, we bound each of the three terms separately. First, since $\vx,\vy\in[0,1]^{dr}$, we have $\norm{\vx}_2\le1$ and $\norm{\vy-\vx}_2\le1$ in the normalized norm, so
\begin{align}
    \frac1{dr}\E\ip{\vx,\vy-\vx}
    \le\E\ab[\norm{\vx}_2\norm{\vy-\vx}_2]\le1.
\end{align}
Next, for the other two terms, Cauchy-Schwarz together with \eqref{eq:multilinear obedience} imply
\begin{align}
    \frac1{dr}\ab|\E\ip{\vx,\vx-\widehat{\vx}}|
    &\le\sqrt{\E\norm{\vx}_2^2}\sqrt{\E\norm{\vx-\widehat{\vx}}_2^2}\le\rho,\\
    \frac1{dr}\ab|\E\ip{\vz^{(j)},\vz^{(j)}-\widehat{\vz}^{(j)}}|
    &\le\sqrt{\E\norm{\vz^{(j)}}_2^2}\sqrt{\E\norm{\vz^{(j)}-\widehat{\vz}^{(j)}}_2^2}\le\rho.
\end{align}
It remains to establish the convex EFP condition. For a convex $1$-Lipschitz function $Q$, \eqref{eq:multilinear obedience} implies
\begin{align}
    \E Q(\vx)
    &\le\E Q(\widehat{\vx})+\E\norm{\vx-\widehat{\vx}}_2
    \le\E Q(\widehat{\vx})+\sqrt{\E\norm{\vx-\widehat{\vx}}_2^2}
    \le\E Q(\widehat{\vx})+\rho.
\end{align}
Next, applying Jensen's inequality,
\begin{align}
    Q(\widehat{\vx}) \le\E\ab[Q\ab(\frac{\vy+M\sum_j\vz^{(j)}}{1+qM})\,\middle|\,\vs_x] \le\E\ab[\frac{Q(\vy)+M\sum_j Q(\vz^{(j)})}{1+qM}\,\middle|\,\vs_x].
\end{align}
Taking expectations,
\begin{align}
    \label{eq:EQ-ineq}
    \E Q(\vx)\le\frac{\E Q(\vy)+M\sum_j\E Q(\vz^{(j)})}{1+qM}+\rho.
\end{align}
Similarly, we obtain
\begin{align}
    \E Q(\vz^{(j)})
    &\le\E Q(\widehat{\vz}^{(j)})+\rho
    = \E Q(\E[\vx\mid\vs_j])+\rho
    \le\E Q(\vx)+\rho,\\
    \E Q(\vy)
    &\le\E Q(\widehat{\vy})+\rho
    \le\E Q(\widetilde g(\vz^{(1)},\dots,\vz^{(q)}))+\rho.
\end{align}
Combining the first inequality with~\eqref{eq:EQ-ineq},
\begin{align}
    (1+qM) \E Q(\vx)
    \le \E Q(\vy) + qM\E Q(\vx)+ (1+2qM)\rho \iff \E Q(\vx) \leq \E Q(\vy) + (1+2qM)\rho.
\end{align}
Using the bound for $\E Q(\vy)$,
\begin{align}
    \E Q(\vx)\le\E Q(\widetilde g(\vz^{(1)},\dots,\vz^{(q)}))+2(1+qM)\rho.
\end{align}
Finally, Lipschitzness of $Q$, \eqref{eq:multilinear copies}, and Cauchy-Schwarz yield
\begin{align}
    \E\ab[Q(\widetilde g(\vz^{(1)},\dots,\vz^{(q)}))-Q(g(\vx))] \le3q\sqrt{\sum_j\E\norm{\vz^{(j)}-\vx}_2^2} \le3q\sqrt{M^{-1}+(2q+1)\rho}.
\end{align}
Combining these bounds gives $\E[Q(\vx)-Q(g(\vx))]\le\eta$.
\end{proof}
In other words, increasing $M$ enforces agreement between the copies, while decreasing $\eps$ makes imitation sufficiently accurate.

\paragraph{Choosing the parameters.} To conclude, we define the parameters as
\begin{align}\label{eq:multilinear precision}
    M=\ceil{64(3q+1)^2/\zeta^2},\quad
    \rho_0=\frac{\zeta}{32(q+1)M},\quad
    \eps=2^{-\ceil{\log_2(4/\rho_0^2)}},
\end{align}
and choose the least positive $m$ such that $4^{-(m+1)}\le\rho_0^2/4$. The choice of $\eps$ guarantees that $\rho\le\rho_0$. The first term in \eqref{eq:multilinear extraction} satisfies
\begin{align}
    2(1+qM)\rho
    \le\frac{\zeta}{16}\frac{1+qM}{(q+1)M}
    \le\frac{\zeta}{16},
\end{align}
where the last inequality uses $M\ge1$. For the second term, since $\zeta\le1$, we have
\begin{align}
    (2q+1)\rho
    \le\frac{(2q+1)\zeta}{32(q+1)M}
    \le\frac1M,
\end{align}
so
\begin{align}
    3q\sqrt{M^{-1}+(2q+1)\rho}
    \le3q\sqrt{\frac2M}
    \le\frac{\sqrt2\,\zeta}{8}\frac{3q}{3q+1}
    \le\frac{\sqrt2\,\zeta}{8}.
\end{align}
Thus, $\eta < \zeta$. By \Cref{cl:multilinear extraction,cl:multilinear amplification}, any such correlated equilibrium supported on at most $2^n$ points yields a $\delta$-fixed point of $f$. Since $\delta^{-1}=O(d)$, $r$ and $b$ are polynomial in $n$, and $\log(1/\zeta)=O(r\log b+\log(r/\delta))=\operatorname{poly}(n)$. As a result, $m$, $\log M$, and $\log(1/\eps)$ are polynomial in $n$. The polytope has $O(drm)$ variables and inequalities, and the utilities contain $O(dr)$ monomials. Thus, $N=\operatorname{poly}(n)$. For every fixed $C$, we have $N^C\le2^n$ for all sufficiently large $n$.
\end{proof}

\subsection{Polynomial precision under PCP-for-PPAD}

\begin{proof}[Proof of \Cref{th:multilinear hardness}(ii)]
We now operate under \Cref{ass:pcp-ppad}. We follow the same construction, with a minor change in the input circuit and the parameters.

We start from the circuit corresponding to \Cref{ass:pcp-ppad}, assumed to have $d$ nodes and encoding length $n \ge d$. We use again the Bernstein approximation with $p\ge(4\lambda/\tau)^2$, and choose a parameter $\delta$ such that $\delta^2\le\gamma\tau^2/64$. If $\norm{f(\vx)-\vx}_2\le\delta$, then
\begin{align}
    \frac1d\big|\{i:|f_i(\vx)-x_i|>\tau/4\}\big|
    \le\frac{16\delta^2}{\tau^2}\le\frac\gamma4.
\end{align}
Every other coordinate satisfies $|h_i(\vx)-x_i|\le\tau/2$ by the Bernstein error bound, so the corresponding gate is satisfied at tolerance $\tau$. The assignment therefore satisfies at least $(1-\gamma)d$ gates and can be decoded in polynomial time using \Cref{ass:pcp-ppad}. This shows that, under~\Cref{ass:pcp-ppad}, the problem is \PPAD-hard even at constant accuracy $\delta = \Theta(1)$.

The rest of the steps are the same. To begin with, we choose $r$ and $b$ per~\eqref{eq:multilinear amplification parameters}, and $\zeta=\delta/(16rb^r)$. Now $b$ is constant, $r=O(\log n)$, and $\zeta=n^{-O(1)}$. Moreover, the parameter choices in~\eqref{eq:multilinear precision} give $M=\Theta(\zeta^{-2})$, $\rho_0=\Theta(\zeta^3)$, $\eps=\Theta(\zeta^6)=n^{-O(1)}$, $m=O(\log n)$. As a result, there are $O(d \log^2 n)$ variables and inequalities and $O(d\log n)$ monomials in each utility.
\end{proof}

\section{Impossibility of fully polynomial-time swap regret minimization}

In this section, we give {\em information-theoretic} lower bounds that preclude the existence of $\poly(d, 1/\eps)$-round algorithms for achieving swap regret $\eps$ in polytope games, even when utilities are linear. Before proceeding, we require some additional prelimiaries.

\paragraph{Tree-form decision problems and extensive-form games.}
Tree-form decision problems describe {\em sequential} interactions between the player and the environment. In a tree-form decision problem, there is a rooted tree $\cT$ of {\em nodes}. There are two types of nodes: {\em decision points}, at which the player selects an action, and {\em observation points}, at which the environment selects an observation. At each decision point $j$, actions are identified with outgoing edges, and we use $A_j$ to denote this set. Leaves of $\cT$ are called {\em terminal nodes}. The dimension $d$ of the strategy set will be the number of terminal nodes.

A {\em pure strategy} is a choice of one action at each decision point. The {\em tree-form} or {\em realization-form representation} of the pure strategy is the vector $\vx \in \{0, 1\}^d$ for which $\vx[z] = 1$ if and only if the player plays {\em all} actions on the path from the root node to terminal node $z$. The set of tree-form pure strategies will be denoted $\Pi \subset \{0, 1\}^d$. We will use $N = |\Pi|$ to denote the number of pure strategies.

Different pure strategies may have the same tree-form representation, but for our purposes we will only require the tree-form representation of strategies, and therefore we will not distinguish between strategies with the same tree-form representation. A {\em mixed strategy} $\mu \in \Delta(\Pi)$ is a distribution over pure strategies.

Tree-form decision problems naturally model the decision problems faced by players in an {\em extensive-form game}. For our purposes, a (perfect-recall) extensive-form game with $n$ players is a polytope game in which utility functions are linear and every player's strategy set is a sequence-form strategy set. It is a well-known result that tree-form decision problems can be represented using polytopes with $\poly(d)$ many linear constraints~\citep{Koller94:Fast,Stengel96:Efficient}.

\begin{remark}
When utility functions are promised to be linear, the restriction to playing distributions over $\Pi$ instead of playing distributions over $\cX := \conv \Pi$ is without loss of generality. This is because, by Carath\'eodory's theorem on convex hulls, a learner can replace any strategy $\vx^{(t)} \in \cX$ with a distribution $\mu^{(t)}$ with expectation $\vx^{(t)}$, and therefore achieve the same expected utility while doing no worse in terms of swap regret.
\end{remark}

\subsection{Prior no-swap-regret algorithms}
\paragraph{The simplex}
The most basic setting for swap regret minimization is the case where $\cX$ is the simplex $\Delta(N)$ on $N$ items, \ie, the game is a {\em normal-form game}. In this setting, accurate upper- and lower bounds are known.

\citet*{Blum07:From} showed that efficient algorithms exist for minimizing swap regret over the simplex.
\begin{theorem}[\citep{Blum07:From}]
    There exists a no-regret learning algorithm for simplices that achieves average swap regret $\eps$ using $T = \tilde O(N/\eps^2)$ rounds.
\end{theorem}

One may wonder whether this is optimal, \eg, whether it is possible to achieve a logarithmic dependence on $N$. Recent simultaneous work by \citet*{Dagan24:From} and \citet*{Peng24:Fast} has essentially completely answered this question for normal-form games.
\begin{theorem}[\citep{Dagan24:From,Peng24:Fast}, upper bound]\label{th:ub}
    There exists a no-regret learning algorithm for simplices that achieves average swap regret $\eps$ within $T = (\log N)^{\tilde O(1/\eps)}$ rounds.
\end{theorem}
Both papers also provided (nearly-)matching lower bounds. Here we state a particularly simple-to-state lower bound proven by \citet*{Dagan24:From}.
\begin{theorem}[Theorem 4.1 of \citep{Dagan24:From}, lower bound]\label{th:lb}
  Let $T < N/4$. Then, in the there exists an oblivious %
  adversary such that the swap regret of any learner for the $N$-simplex is $\Omega(\log^{-5}T)$.
\end{theorem}

\paragraph{Extensive-form games}

For extensive-form games, the picture is less clear. For an upper bound, one can consider a tree-form decision problem with $N$ pure strategies (\ie, $|\Pi| = N$) as simply an ``easier version'' of a normal-form decision problem where each pure strategy is treated as a different action, \ie, where the strategy set is the $N$-simplex. \Cref{th:ub} therefore implies a similar bound on swap regret for tree-form decision problems.\footnote{As stated, the bound is only information-theoretic. However, the information-theoretic bound is implementable by an efficient (\ie, $\poly(m, 1/\eps)$-time-per-iteration) algorithm, which is described by \citet*{Dagan24:From} and \citet*{Peng24:Fast}, and is beyond the scope of this paper.}
\begin{corollary}[\citealp{Dagan24:From,Peng24:Fast}, tree-form upper bound]\label{th:efg-upper}
    There exists a no-regret learning algorithm for tree-form decision problems that achieves swap regret $\eps$ after $T = (\log N)^{\tilde O(1/\eps)} \le d^{\tilde O(1/\eps)}$ rounds.
\end{corollary}

Showing a matching lower bound for extensive form, however, remained open. The main difficulty is that the adversary is restricted to {\em linear} utility functions $u^{(t)} : \cX  \to \R$; the adversary in \Cref{th:lb} does not use linear utility functions when the extensive-form game is interpreted as a normal-form game over $N$ actions as described above. The main result in this section closes this discrepancy, by showing a lower bound that almost matches \Cref{th:efg-upper}.

The main result of this section is the following.

\begin{theorem}\label{th:efg-lower}
    There exist arbitrarily large tree-from strategy sets with the following property. Let $\eps > 0$ and suppose $T \le \exp\ab(\Omega\ab(\min\ab\{d^{1/14}, \eps^{-1/6}\}))$. Then there exists an oblivious adversary running for $T$ iterations against which no learner can achieve expected swap regret better than $\eps$.
\end{theorem}
Intuitively, the proof of \Cref{th:efg-lower} works by finding an ``embedding'' of the adversary of \Cref{th:lb} into a tree-form decision problem such that the utility functions $u^{(t)}$ do remain linear. This works by choosing random vectors in $\{-1, 1\}^n$ (for some appropriately chosen dimension $n$) to simulate the ``actions'' in the (exponentially large) normal-form decision problem, and then exploiting the concentration property that an exponentially large number of such vectors $\{\va_i\}_{i=1}^{N}$ can be chosen such that $\ip{\va_i, \va_j} \approx 0$ for all $i \ne j$.

Before proving \Cref{th:efg-lower}, we first state a more detailed version of the normal-form lower bound (\Cref{th:lb}).\footnote{The discussion of \citet*[pp 37--38]{Dagan24:From} specifies the adversary which satisfies the properties listed in \Cref{th:normalform}.} This restatement changes the notation so as to avoid mixing the notation between tree-form and normal-form decision problems, and extracts some useful properties of the adversary. In particular, a key property of the adversary that we exploit (specified in \Cref{it:adv-unit} in the below theorem) is that, with probability 1, each of the utilities that it chooses is in fact the indicator of a single action. (In general, in the normal-form game setting, the vectors $\vu^{(t)}$ may have coordinates in $[-1,1]$.)

\begin{theorem}[\citep{Dagan24:From}, expanded version of \Cref{th:lb}]\label{th:normalform}
    Let $\cA$ be the set of vertices of the $N$-simplex, and let $T < N/4$. Then there exists an adversary on $\cA$ with the following properties:
    \begin{enumerate}
        \item  \label{it:adv-unit} The adversary selects a sequence $(u^{(1)}, \dots, u^{(T)}) \sim \cD$ from some distribution $\cD \in \Delta(\cA^T)$, and then outputs utility  $\ind\{ a^{(t)} = u^{(t)}\} $ at time $t$ regardless of the sequence of distributions played by the learner.
        \item There exists an action $a^* \in \cA$ that is never used by the adversary.\footnote{This can always be assumed WLOG.}
        \item There exists a partition $\cA = \cA_1 \sqcup \dots \sqcup \cA_k$ where $k \le O(\log T)$ with the following property.
        Within each set $\cA_i$, number the actions $\cA_i = \{ a_{i1}, \dots, a_{iN_i} \}$. For any sequence $(u^1, \dots, u^T)\in \supp \cD$, the adversary plays actions in $\cA_i$ only in increasing order. That is, if $u^{(t)} = a_{ij}$ and $u^{(t')} = a_{ij'}$ and $t \le t'$, then $j \le j'$.
        \item The swap regret of any learner against this adversary is\footnote{The reason that the $-5$ in \Cref{th:lb} has been changed to a $-6$ here is that the adversary is now constrained to pick a sequence of {\em actions}, \ie, $\ell_1$-bounded losses, instead of $\ell_\infty$-bounded losses. See \citet*[Theorems 1.7 and 4.1]{Dagan24:From}.} $\Omega(k^{-6}) = \Omega(\log^{-6} T)$. %
        \label{prop:increasing}
    \end{enumerate}
\end{theorem}

We now prove \Cref{th:efg-lower} via a reduction from \Cref{th:normalform}. In particular, we fix $N \in \mathbb{N}$ as in \Cref{th:normalform}, and let $\cA$ denote the action set for the lower bound in \Cref{th:normalform}. We will often use the decomposition $\cA = \cA_1 \sqcup \cdots \sqcup \cA_k$.

\subsection{The adversary}

We now specify the adversary that we will use in the proof. Consider the following family of tree-from strategy sets, parameterized by natural numbers $k$ and $n$. First the learner picks an index $i \in [k]$. Then the environment picks $j \in [n]$, and finally the learner picks a binary action. This family of decision problems is depicted in \Cref{fig:dp}.

\begin{figure}[h]
\definecolor{p1color}{RGB}{31,119,180}
\newcommand{\pone}{{\ensuremath{\color{p1color}\blacktriangle}}\xspace}
\tikzset{
  every path/.style={-},
}
\forestset{
        default preamble={for tree={
        parent anchor=south, child anchor=north
}},
  p1/.style={
      regular polygon,
      regular polygon sides=3,
      inner sep=2pt, fill=p1color, draw=p1color,
  },
  nat/.style={draw},
  terminal/.style={draw=none, inner sep=2pt},
   el/.style n args={1}{edge label={node[midway, fill=white, inner sep=1pt, draw=none, font=\footnotesize] {#1}}},
}
\centering
\scalebox{0.8}{
\begin{forest}
for tree={l+=1em},
    [,p1
        [,nat,el={$i{=}1$},s sep=3em
            [,p1,el={$j{=}1$} [] []]
            [,p1,el={$j{=}2$} [] []]
            [$\dots$, draw=none, no edge]
            [,p1,el={$j{=}n$} [] []]
        ]
        [,nat,el={$i{=}2$},s sep=3em
            [,p1,el={$j{=}1$} [] []]
            [,p1,el={$j{=}2$} [] []]
            [$\dots$, draw=none, no edge]
            [,p1,el={$j{=}n$} [] []]
        ]
        [$\dots$, draw=none, no edge]
        [,nat,el={$i{=}k$},s sep=3em
            [,p1,el={$j{=}1$} [] []]
            [,p1,el={$j{=}2$} [] []]
            [$\dots$, draw=none, no edge]
            [,p1,el={$j{=}n$} [] []]
        ]
    ]
\end{forest}
}
    \caption{A depiction of the class of tree-form decision problems used in the proof of \Cref{th:efg-lower}. Triangles (\pone) are decision points and boxes ($\square$) are observation points. }\label{fig:dp}
\end{figure}

A pure strategy is identified (up to linear transformations) by a vector $\vx \in \R^{k\times n}$ where, for some $i \in [k]$, $\vx[i, \cdot] \in \{-\frac{1}{\sqrt{n}}, \frac{1}{\sqrt{n}} \}^n$ and $\vx[i', \cdot]=\vec 0$ if $i' \ne i$ (\ie, $\vx$ interpreted as a matrix with exactly one nonzero row). For convenience, we will use $\Pi_i \subset \Pi$ to denote the set of (pure) strategies where the learner plays $i$ at the root. Let $C$ be an absolute constant large enough to make the asymptotic bounds in \Cref{th:normalform} true.

The adversary used to prove \Cref{th:efg-lower} works as follows. First, for each $i \in [k]$, it populates $\cA_i$ with $N_i$ uniformly randomly chosen strategies in $\Pi_i$. %
Formally, we let $\psi : \cA \to \Pi$ denote the (random) mapping which associates each action in $\cA_i \subset \cA$ with the corresponding action in $\Pi_i$: the image of $\cA_i$ under $\psi$ consists of actions we denote by $\tilde\va_{i1}, \ldots, \tilde\va_{iN_i} \in \Pi_i$, which are chosen independently and uniformly in $\Pi_i$. The adversary in \Cref{th:normalform} produces a random sequence $u^{(1)}, \ldots, u^{(T)} \in \cA$; we consider the adversary which draws a sequence of utilities $u^{(t)}$ from that distribution and outputs the sequence consisting of $\tilde \vu^{(t)} := \psi(u^{(t)})$ for $t \in [T]$. %

\subsection{Analysis}

Let $\eps > 0$ be a parameter to be selected later. We start with a simple concetration bound.

\begin{lemma}\label{lem:concentration}
    Let $\delta = N^2 e^{-n\eps^2/2}$. With probability $1 - \delta$, for all $a,a ' \in \cA$, we have $\abs{\ip{\tilde{\vec a}, \tilde{\vec a}'} - \ind\{a = a'\}} \le \eps$.
\end{lemma}
\begin{proof}
    If $a = a '$ then the claim holds trivially because then $\tilde{\va} = \tilde{\va}'$, and they are both unit vectors. For a fixed $a \ne a' \in \cA$, the claim holds with probability $2e^{-n\eps^2/2}$ by Hoeffding's inequality. The lemma then follows by union bounding over the $\binom{N}{2} \le N^2/2$ pairs.
\end{proof}

We will claim that, for any learner against this adversary, there exists a learner against the adversary of \Cref{th:normalform} that achieves a similar swap regret---and thus the swap regret of the former learner must be large. First, we will construct the latter learner.

Let $\mu^{(1)}, \dots, \mu^{(T)} \in \Delta(\Pi)$ be the sequence of distributions played by the learner. Note that $\mu^{(t)}$ can depend on the utilities $u^{(1:t-1)} \in \cA$ that are played by the adversary. Consider the sequence $\bar\mu^{(1)}, \dots, \bar\mu^{(T)} \in \Delta(\cA)$, where $\bar\mu^{(t)}$ is the distribution that samples $\vx \sim \mu^{(t)}$ and plays according to $p_\vx \in \Delta(\cA)$, defined as follows. Let $\vx \in \Pi_i$ be any strategy. There are two cases.

\begin{enumerate}
    \item
    $\ip{\vx,\tilde \va_{ij}} \le \eps$ for every $i \in [k], j \in [N_i]$. Then define $p_\vx= a^*$ deterministically (i.e., $p_\vx$ is the distribution which puts all of its mass on $a^*$).
    \item $\ip{\vx, \tilde\va_{ij}} > \eps$ for some $i \in [k], j \in [N_i]$. Let $j$ be the {\em largest} such index, let $\beta = \ip{\vx, \tilde\va_{ij}}$, and define $p_\vx$ as the distribution that is $a^*$ with probability $1-\beta$ and $a_{ij}$ with probability $\beta$. Note that $\beta \in [0,1]$ since $\vx, \tilde\va_{ij} \in \{ -\frac{1}{\sqrt n}, \frac{1}{\sqrt n} \}^n$ and in this case we have assumed that $\langle \vx, \tilde\va_{ij} \rangle > \eps > 0$.
\end{enumerate}
A critical property for us will be that the learner cannot ``guess in advance'' what future unobserved $\tilde\va_{ij}$s will be, since these are sampled uniformly at random. That is, in Case 2, $\vx$ can only be played with large probability once the adversary has played $\tilde\va_{ij}$.

To be more formal, we first define some notation. For every $i \in [k], j \in [N_i]$, let $t_{ij}$ be the first iteration on which the  adversary plays $\tilde\va_{ij}$ (or $t_{ij} = T$ if this never happens). For $\vx \in \Pi_i$, if $\vx$ is in Case 1 above then define $t_\vx = 0$, and otherwise define $t_\vx = t_{ij}$, where $j$ is as in Case 2.

There are two properties that we will critically need to use about $t_\vx$. The first states that the learner cannot place large mass on $\vx$ until after $t_\vx$, because doing so would require the learner to guess a vector heavily correlated with $\tilde\va_{ij}$ before the learner observes $\tilde\va_{ij}$.

\begin{lemma}
  \label{lem:delta-bound} $\ds \E\frac{1}{T}\sum_{\vx \in \Pi} \sum_{t =1}^{t_\vx} \mu^{(t)}(\vx) \le \delta$.
\end{lemma}
\begin{proof}
    Since the learner has not yet observed $\tilde\va_{ij}$ at time $t_{ij}$, its prior strategy sequence $\mu^{(1:t_{ij})}(\vx)$ must be independent of $\tilde\va_{ij}$. Moreover, if $t \le t_\vx$ then there must exist some $j$ with $t_{ij} \ge t$ and $\ip{\vx, \tilde\va_{ij}} \ge \eps$---namely, the $j$ defining Case 2. Thus we have:
    \begin{align}
        \E \frac{1}{T} \sum_{\vx \in \Pi} \sum_{t=1}^{t_\vx}  \mu^{(t)}(\vx)
        &\le  \E \frac{1}{T} \sum_{i=1}^k \sum_{\vx \in \Pi_i} \sum_{t=1}^T   \mu^{(t)}(\vx) \sum_{j : t_{ij} \ge t} \ind\{\ip{\vx, \tilde\va_{ij}} \ge \eps\}
        \\&=  \underbrace{\frac{1}{T} \sum_{i=1}^k \sum_{\vx \in \Pi_i} \sum_{j=1}^{N_i} \E\ab\Big[  \sum_{t \le t_{ij}} \mu^{(t)}(\vx) ]}_{\le N} \underbrace{ \E \ab[ \ind{\ip{\vx,\tilde \va_{ij}} \ge \eps} ] }_{\le e^{-n\eps^2/2}} \le \delta,
    \end{align}
    where in the last line we use the fact that $\tilde\va_{ij}$ is independent of $\mu^{1:t_{ij}}(\vx)$ and then Hoeffding's inequality. Moreover, we have used in the final inequality that for each $i \in [k]$ and $j \leq N_i \leq N$, we have $\frac 1T \sum_{\vx \in \Pi}\sum_{t=1}^T \E\mu^{(t)}(\vx) \leq 1$.
\end{proof}

The second property is that, for  $t > t_\vx$, utilities of $\vx$ under $\tilde\vu^{(t)}$ are approximately the same as those of $p_\vx$ under the utilities $\vu^{(t)}$ of \Cref{th:normalform}.

\begin{lemma}
  \label{lem:epsilon-bound}
    For $t > t_\vx$, we have $\ds \ip{\vx, \tilde\vu^{(t)}} \le p_\vx(u^{(t)}) + \eps$.
\end{lemma}
\begin{proof}
Let $\vx \in \Pi_i$. There are two cases. In the first case, we suppose that $\ip{\vx, \tilde\va_{ij}} \le \eps$ for every $i \in [k], j \in [N_i]$. Then for every $t$, we have $\vu^{(t)} \notin \supp p_\vx = \{ a^* \}$ (because the adversary of \Cref{th:normalform} never plays $a^*$), and $\ip{\vx, \tilde\vu^{(t)}} \le \eps$ by definition, so we are done.

Otherwise, let $j$ be the largest index for which $\ip{\vx, \tilde\va_{ij}} > \eps$. Then $t_\vx = t_{ij}$ by definition, and since $t > t_{ij}$, by Property~\ref{prop:increasing}, for time steps following $t_{ij}$ the adversary of \Cref{th:normalform} is no longer allowed to play $a_{ij'}$ for $j' < j$. Thus, either $u^{(t)} = a_{ij}$, or else $u^{(t)} \notin \{ a_{i1}, \ldots, a_{ij} \}$. Since $j$ is defined to be the largest index for which $\langle \vx, \tilde\va_{ij} \rangle > \eps$, in the latter case %
$\langle \vx, \tilde\vu^{(t)} \rangle \leq \eps \le p_\vx(u^{(t)}) + \eps$. In the former case, we have $\langle \vx, \tilde\vu^{(t)} \rangle = \beta = p_\vx(a_{ij})= p_\vx(u^{(t)})$.
\end{proof}

For the rest of this proof we will use $\bar V(\phi)$ to denote the utilities experienced by $\bar\mu^{(t)}$ under the utilities $\vu^{(t)}$ in \Cref{th:normalform} and deviation function $\phi : \cA \to \cA$. That is,
\begin{align}
    \bar V(\phi) = \frac1T \sum_{t=1}^T \sum_{\va \in \cA} \bar\mu^{(t)}(a) \ind\{\phi(a) = u^{(t)}\} = \frac1T \sum_{t=1}^T \sum_{\vx \in \Pi} \mu^{(t)}(\vx) \Pr_{a \sim p_\vx}[\phi(a) = u^{(t)}]
\end{align}
Similarly, we will use
\begin{align}
    V(\phi) := \frac{1}{T} \sum_{t=1}^T \E_{\vx \sim \mu^{(t)}} \ip{\vu^{(t)}, \phi(\vx)}
\end{align}
to denote the utility experienced by $\mu^{(t)}$ under our extensive-form adversary.

By \Cref{th:normalform}, there exists a function $\bar\phi : \cA \to \cA$ such that\footnote{Technically $\phi$ is a random variable dependent on $\vu^{(1)}, \dots, \vu^{(T)}$.} $\E[\bar V(\bar\phi) - \bar V(\Id)] \ge 1/Ck^6$.
We define a deviation $\phi : \Pi \to \cX$ by setting\footnote{Note that if a profitable deviation $\Pi \to \cX$ exists, then by linearity, so must a profitable deviation $\Pi \to \Pi$.} $\phi(\vx) := \E_{a \sim p_\vx} \psi(\bar\phi(a)).$

It suffices to show that $\E[V(\phi) - V(\Id)]$ is large. To do this, we will show that, in expectation and up to small errors, $V(\Id) \le \bar V(\Id)$ and $V(\phi) \ge \bar V(\bar\phi)$.

For the first approximation, we have
\begin{align}
    V(\Id) &= \frac{1}{T} \sum_{\vx \in \Pi} \sum_{t=1}^T \mu^{(t)}(\vx) \ip{\vx, \tilde \vu^{(t)}}
    \\&\le \frac{1}{T} \sum_{\vx \in \Pi} \sum_{t > t_\vx} \mu^{(t)}(\vx) \ip{\vx, \tilde \vu^{(t)}} + \delta
    \\&\le \frac{1}{T} \sum_{\vx \in \Pi} \sum_{t > t_\vx} \mu^{(t)}(\vx) p_\vx(u^{(t)}) + \eps + \delta
    \\&\le \frac{1}{T} \sum_{\vx \in \Pi} \sum_{t=1}^T \mu^{(t)}(\vx) p_\vx(u^{(t)}) + \eps + 2 \delta
    \\&= \bar V(\Id) + \eps + 2 \delta,\label{eq:vid}
\end{align}
where the first and third inequalities use \Cref{lem:delta-bound}, and the second inequality uses \Cref{lem:epsilon-bound}.

For the second, let $F$ be the event in \Cref{lem:concentration}. Conditional on $F$, we have
\begin{align}
V(\phi) &= \frac{1}{T} \sum_{\vx \in \Pi} \sum_{t=1}^T \mu^{(t)}(\vx) \ip{\phi(\vx), \tilde\vu^{(t)}}
\\&\ge \frac{1}{T} \sum_{\vx \in \Pi} \sum_{t > t_\vx} \mu^{(t)}(\vx) \ip{ \phi(\vx), \tilde\vu^{(t)}} - \delta
\\&\ge \frac{1}{T} \sum_{\vx \in \Pi} \sum_{t > t_\vx} \mu^{(t)}(\vx) \Pr_{a \sim p_\vx}[\bar\phi(a) = u^{(t)}] - \eps - \delta
\\&\ge  \sum_{\vx \in \Pi} \frac{1}{T} \sum_{t=1}^T \mu^{(t)}(\vx) \Pr_{a \sim p_\vx} [\bar \phi(a) = u^{(t)}]- \eps - 2\delta
\\&= \bar V(\bar\phi) - \eps - 2\delta,
\end{align}
where the first and third inequalities use \Cref{lem:delta-bound}.  To establish the second inequality above, we note that,
\begin{align}
    \ip{ \phi(\vx), \tilde\vu^t} = \E_{a \sim p_\vx} \ip{\psi(\bar\phi(a)), \tilde\vu^t} \ge \Pr_{a \sim p_\vx}[\bar\phi(a) = u^t] - \eps
\end{align}
by \Cref{lem:concentration}, since $\bar\phi(\va), \tilde\vu^t \in \cA$.

Thus, accounting for the probability $\delta$ in which $F$ fails, we have
\begin{align}
    \E[V(\phi) - \bar V(\bar\phi)] = \underbrace{\E[V(\phi) - \bar V(\bar\phi)|F]}_{\ge -\eps - 2\delta} \cdot \underbrace{\Pr[F]}_{\le 1} + \underbrace{ \E[V(\phi) - \bar V(\bar\phi) | \neg F] }_{\ge -1} \cdot \underbrace{\Pr[\neg F]}_{\le \delta} \ge - \eps - 3\delta\label{eq:vphi}.
\end{align}
Combining \eqref{eq:vid} and \eqref{eq:vphi},
\begin{align}
  \E[V(\phi) - V(\Id)] \ge \E\ab[\bar V(\bar\phi) - \bar V(\Id)] - 2 \eps - 5 \delta \ge \frac{1}{Ck^6} - 2 \eps - 5 \delta \ge \frac{1}{4Ck^6} = \eps\label{eq:vphi-vid}
\end{align}
by setting the parameters
\begin{align}
    \eps = \frac{1}{4Ck^6} \qq{and}
    n = \frac{2\log 20CN^2 k^6}{\eps^2} \qq{so that}  5\delta = 5N^2e^{-n\eps^2/2} \le \frac{1}{4Ck^6}. \qedhere
\end{align}

The resulting tree-form decision problem hence has dimension $d = k \cdot n = O(\log^{14} N)$, and since $k = \Theta(\eps^{-{1/6}}) \le O(\log T)$ we have that the swap regret is at least $\eps$ for all $T < \min\ab\{N/4, \exp(\Omega(\eps^{-1/6}))\}$, where $N = \exp(\Omega(d^{1/14}))$, completing the proof of \Cref{th:efg-lower}. \qed

\section{Poly-dimensional equilibrium learning and computation}\label{sec:poly learning}

In this section, we develop algorithms for learning and computing equilibria in concave games when the set $\Phi$ is poly-dimensional.   

\subsection{Ellipsoid against hope and expected fixed points}

\citet*{Daskalakis25:Efficient} showed that, when the utilities of each player are {\em linear}, $\Phi$-equilibria are efficiently computable. However, their techniques fundamentally depended on the linearity of the utilities. Here, we first review their algorithm informally and why it fails for concave utilities, and then we will describe our algorithm. Consider the problem of finding an approximate $\Phi$-equilibrium\footnote{Here, we use a slightly different notion of $\eps$ from the prelimiaries, summing the deviation benefits instead of taking a maximum. This does not matter because 1) we have assumed that $\Id_i \in \Phi_i$ for all $i$, and 2) we are not concerned with bounding the polynomial factors, so an extra factor of $n$ is immaterial.}
\begin{align}
    \qq{find} \mu \in \Delta(\cX) \qq{s.t.} \sum_{i=1}^n \E_{\vx\sim\mu} \ab[ u_i(\phi_i(\vx_{i}), \vx_{-i}) - u_i(\vx)] \le \eps \quad \forall (\phi_1, \dots, \phi_n) \in \bigtimes_{i=1}^n \Phi_{\vm_i} \label{eq:primal} \tag{EAH-P}
\end{align}
and its dual
\begin{align}
    \qq{find} (\phi_1, \dots, \phi_n) \in \bigtimes_{i=1}^n \Phi_{\vm_i}  \qq{s.t.} \sum_{i=1}^n u_i(\phi_i(\vx_{i}), \vx_{-i}) - u_i(\vx) > \eps \quad \forall \vx \in \cX.
\end{align}
Or, equivalently, taking advantage of the low-dimensional representation of $\Phi_{\vm_i}$:
\begin{align}
    \qq{find} (\mK_1, \dots, \mK_n)  \qq{s.t.} \\\sum_{i=1}^n u_i(\mK_i \vm_i(\vx_i), \vx_{-i}) - u_i(\vx) \ge \eps \quad& \forall \vx \in \cX \label{eq:dual} \tag{EAH-D} \\
    \mK_i\vm_i(\vx_i) \in \cX_i \quad& \forall i \in [n], \vx_i \in \cX_i.
\end{align}
The two problems are dual in the sense that \eqref{eq:primal} is always feasible and \eqref{eq:dual} is always infeasible, and a certificate of infeasibility of the dual problem gives an $\eps$-approximate $\Phi$-equilibrium. Therefore, the algorithm works by running the ellipsoid algorithm on \eqref{eq:dual}, eventually terminating due to the infeasibility, and extracting a solution from the certificate of infeasibility. This framework is commonly referred to as {\em ellipsoid against hope} (EAH)~\citep{Farina24:Polynomial,Zhang25:Learning}; indeed, it is a generalization of the algorithm of that name~\citep{Papadimitriou08:Computing} that was originally devised for many-player normal-form games. To run the ellipsoid algorithm, the following subproblem, commonly referred to as the {\em semi-separation} problem~\citep{Daskalakis25:Efficient,Zhang25:Learning}: given a function $\phi_i : \cX_i \to \R^{d_i}$, {\em either}
\begin{enumerate}
    \item produce a distribution $\mu_i \in \Delta(\cX_i)$ such that $\phi_i$ will be no more than $\eps/2$-profitable regardless of the strategies of other players, {\em or}
    \item produce a certificate that $\phi_i$ is not an endomorphism, \ie, produce a $\vx_i$ such that $\phi_i(\vx_i) \notin \cX_i$. 
\end{enumerate}
If any semi-separation oracle produces (2), we can use the violating $\vx_i$ to find a separating hyperplane for the constraint $\mK_i \vm_i(\vx_i) \in \cX_i$. Otherwise, the distribution $\mu := \mu_1 \times \dots \times \mu_n$, which samples each $\vx_i \sim \mu_i$ independently, satisfies $\sum_{i=1}^n u_i(\mK_i \vm_i(\vx_i), \vx_{-i}) - u_i(\vx) \le \eps/2$ and therefore can once again be used to construct a separating direction. 

The semi-separation problem is solved by {\em another} instantiation of EAH. Consider the problem of finding a distribution such that $\phi_i$ will not be $\eps$-profitable:
\begin{align}
    \qq{find} \mu_i \in \Delta(\cX_i) \qq{s.t.} \E_{\vx_i\sim\mu_i} \ip{\vv, \phi_i(\vx_i) - \vx_i} \le \eps \quad \forall \vv \in B(\vec 0, 1). \label{eq:efp-primal} \tag{EFP}
\end{align}
This problem is known as the {\em expected fixed point} problem~\citep{Zhang24:Efficient}, because it is equivalent to finding a distribution with $\E \phi(\vx_i) \approx \E \vx_i$. 
The problem \eqref{eq:efp-primal} is solved, once again, by taking a dual.
\begin{align}
    \qq{find} \vv \in B(\vec 0, 1) \qq{s.t.} \ip{\vv, \phi_i(\vx_i) - \vx_i} > \eps \quad \forall \vx_i \in \cX_i \label{eq:efp-dual} \tag{EFP-D}
\end{align}
Once again, \eqref{eq:efp-dual} is infeasible, and a certificate of infeasibility yields a solution to \eqref{eq:efp-primal}. It remains to demonstrate a separation oracle for \eqref{eq:efp-dual}. Given $\vv \in [-1, 1]^{d_i}$, consider $\vx = \argmax_{\vx_i \in \cX_i} \ip{\vv, \vx_i}$. There are two cases. First, if $\phi(\vx_i)$ is inside $\cX_i$, then we must have $\ip{\vv, \phi(\vx_i) - \vx_i} \le 0$, so $\phi(\vx_i) - \vx_i$ is a separating direction. Otherwise, if $\phi(\vx_i)$ is outside $\cX_i$, then we can terminate because we have proven that $\phi_i$ is not an endomorphism.

The assumption that utilities are linear becomes important here: to show that $\phi_i$ is not profitable against any utility $u$, we need $\vv$ to be the gradient of the utility function, but we also need $\vv$ to not depend on the strategy $\vx_i$---that is, we need $u_i$ to be linear in $\vx_i$. To go beyond linear functions, we would need to replace $\ip{\vv, \vx}$ with an arbitrary concave function. That is, we hope to solve the problem
\begin{align}
    \qq{find} \mu_i \in \Delta(\cX_i) \qq{s.t.} \E_{\vx_i \sim\mu_i} \ab[u(\phi_i(\vx_i)) - u(\vx_i)] \le \eps \quad \forall \text{concave } u : \cX_i \to [0, 1].
\end{align}
But this is exactly the convex expected fixed point problem, which we already know is \PPAD-hard under the support and accuracy bounds of~\Cref{th:fixed set hardness}!

\subsection{Hardness of $\Phi$-equilibria in concave games}

Our first result is negative: the hope that we lay out above is not possible, at least if our target is $\polylog(1/\eps)$-time algorithms and contraction fixed points are hard. Indeed, in this section, we will show that the problem of finding a $\Phi$-equilibrium is \Contr-hard.

\begin{theorem}
    \label{thm:polydimcontract}
    Finding an $\eps$-approximate $\Phi$-equilibrium is \Contr-hard, even for two-player concave games where both players share the same strategy set $\cX$ and $\Phi_1 = \Phi_2 = \{ f\}$ is a singleton. \label{th:concave efg hard}
\end{theorem}
\begin{proof}
    Let $f : \cX \to \cX$ be a $(1-\gamma)$-contraction under some norm $\norm{\cdot}$. Define a two-player game as follows. Both players' strategy sets are $\cX$, and the utilities are given by
    \begin{align}
    u_1(\vx, \vy) = -\norm{\vx - \vy}, \qq{} u_2(\vx, \vy) = -\norm{\vy - f(\vx)}.
    \end{align}
    Let $\mu$ be a $\eps$-approximate $\Phi$-equilibrium with $\Phi_1 = \Phi_2 = \{f\}$. Summing the two players' deviation benefits under $f$, we have
    \begin{align}
       2 \eps &\ge \E \ab[u_1(f(\vx), \vy) - u_1(\vx, \vy) + u_2(\vx, f(\vy)) - u_2(\vx, \vy)]
        \\&= \E \ab[\norm{\vx - \vy} - \norm{f(\vx) - f(\vy)}]
        \\&\ge \gamma \E \norm{\vx - \vy}.
    \end{align}
    where all expectations are over $(\vx, \vy) \sim \mu$. Moreover,
    \begin{align}
        \eps \ge \E [u_2(\vx, f(\vy)) - u_2(\vx, \vy)] &= \E [\norm{\vy - f(\vx)} - \norm{f(\vy) - f(\vx)}]
        \\&\ge \E [\norm{\vy - f(\vy)} - 2 \norm{f(\vy) - f(\vx)}]
        \\&\ge \E \ab[\norm{\vy - f(\vy)} - 2\norm{\vx - \vy}]
        \\&\ge \E \norm{\vy - f(\vy)} - \frac{4\eps}{\gamma}
    \end{align}
    so $\E \norm{\vy - f(\vy)} \le \eps(1 + 4/\gamma) \le 4\eps/\gamma$. Thus, taking $\eps = \delta\gamma/4$ solves \Contr.
\end{proof}

\subsection{A $\polylog(1/\eps)$-time algorithm for low-dimensional $\Phi$-equilibria with poly-dimensional concave utilities}\label{sec:qefp}

While the general $\Phi$-equilibrium and CEFP problems are \Contr-hard and \PPAD-hard respectively, there are cases where the $\Contr$ problem is easy. Indeed, $\Contr$ is easy when we restrict the norm $\norm{\cdot}$ to be the $\ell_p$-norm, with $p$ an even integer~\cite{Daskalakis27:Polynomial}. What is the generalization of this easy case to CEFP and concave equilibria? In this subsection, we answer this question.

We start with CEFPs, and again we will drop player indices for the remainder of this section. Following \citet*{Daskalakis27:Polynomial}, given a function $f : \cX \to \cX$, given a set $\cG$ of operators $\cX \to \R^d$,  a distribution $\mu \in \Delta(\cX)$ is an $\eps$-approximate {\em $\cG$-expected fixed point} ($\cG$-EFP) if it satisfies
    \begin{align}
        \E_{\vx\sim\mu} \ip{G(\vx), \vx - f(\vx)} \le \eps \qq{for all} G \in \cG. \label{eq:gefp} \tag{$\cG$-EFP}
    \end{align}
Note that, in particular, an $\eps$-approximate $\cG$-EFP satisfies the inequality in~\eqref{eq:cefp-primal} for every convex $Q$ with $\grad Q \in \cG$.
\citet*{Daskalakis27:Polynomial} showed that $\cG$-expected fixed points are efficiently computable so long as $\cG$ itself has a nice representation:
\begin{theorem}[Theorem~3.3 of \citet*{Daskalakis27:Polynomial}]\label{th:gefp}
    Given $L$-Lipschitz, $L$-bounded feature map $m : \cX \to \R^k$, well-described set $\cX$, and $f : \cX \to \cX$, there is a $\poly(d, k, \log(RL/\eps))$-time algorithm that computes an $\eps$-approximate $\cG(m)$-EFP, where $$\cG(m) := \{ G : \cX \to \R^d\ | \ G(\vx) = \mA m(\vx),\ \mA\in\R^{d\times k},\ \norm{\mA}_2 \le L,\ G \textup{ monotone} \}.$$
\end{theorem}

Thus, by using the $\cG$-EFP oracle as a drop-in replacement for the expected fixed point oracle in the framework of \citet*{Zhang25:Learning}, it follows that $\Phi$-equilibria can be computed efficiently if the {\em utilities} are efficiently representable:
\begin{theorem}\label{th:concave quadratic eqcomp}
    There is a $\poly(d, k, \log(RL/\eps))$-time algorithm for computing an $\eps$-approximate $\Phi_\vm$-equilibrium of a given concave game with a constant number of players, under the following assumption: for every player $i$ we are given a map $m_i : \cX_i \to \R^k$ such that, for every profile $\vx_{-i}$ for the other players, the map $\vx_i \mapsto -\grad_{\vx_i} u(\vx_i, \vx_{-i})$ is in $\cG(m_i)$.
\end{theorem}
\begin{remark}
The algorithm suggested by the above theorem is a {\em triply}-nested ellipsoid algorithm! The three layers are the following.
\begin{itemize}
    \item The outer layer solves the actual equilibrium computation problem \eqref{eq:primal}.
    \item The middle layer, given $\phi_i : \cX_i \to \cX_i$ (and now dropping the index $i$ for cleanliness of notation), solves \eqref{eq:gefp} with $f = \phi_i$.
    \item The inner layer, as per \citet{Daskalakis27:Polynomial}, is called as a subroutine within the $\cG$-EFP oracle to perform the following semi-separation: given $G : \cX \to \R^d$, either solve the monotone variational inequality associated with $G$, that is, 
    \begin{align}
        \qq{find} \vx \in \cX \qq{s.t.} \ip{G(\vx), \vx - \vy} \le \eps \qq{for all} \vy \in \cX
    \end{align}
    or derive a cut showing that $G$ is not monotone, that is, 
    \begin{align}
        \qq{find} \vx, \vy \in \cX \qq{s.t.} \ip{G(\vx) - G(\vy), \vx - \vy} < 0.
    \end{align}
\end{itemize}
\end{remark}
\begin{remark}
In particular, the above gives a $\poly(d^p, \log(RL/\eps))$-time algorithm for games in which all utilities are polynomials of degree at most $p$.
\end{remark}

\begin{remark}
Unlike \citet*{Daskalakis27:Polynomial}, the previous remark does not require that $p$ is an even integer; \citet*{Daskalakis27:Polynomial} use the evenness of $p$ only in their reduction to contraction maps, which we do not need here.
\end{remark}

\begin{remark}
The assumption that the number of players is constant is crucial to this analysis: without it, it is not even possible to compute the expected utility of the separating distribution $\mu = \mu_1 \times \dots \times \mu_n$, because it has support exponential in $n$. \citet*{Zhang25:Learning} circumvent this issue by exploiting multilinearity; we, of course, cannot do that.
\end{remark}

\subsection{FPTAS for CEFPs and low-dimensional $\Phi$-regret minimization for arbitrary concave functions}\label{sec:low dim fptas}
We end this section by observing that there exists an FPTAS for CEFPs in arbitrary concave functions, and therefore, again following the framework of \citet*{Zhang25:Learning}, an efficient algorithm for {\em regret minimization} with concave utilities and low-dimensional $\Phi$ which in turn immediately implies an FPTAS for computing low-dimensional $\Phi$-equilibria in general concave games.

\begin{proposition}\label{prop:fptas cefp}
    There is a $\poly(d, 1/\eps)$-time algorithm that, given $\phi : \cX \to \R^d$, either 1) solves the CEFP problem with respect to $\phi$ or 2) finds some $\vx$ such that $\phi(\vx) \notin \cX$. 
\end{proposition}

\begin{proof}
    Let $\vx^{(0)}$ be arbitrary, and for $j = 1, \dots, M \defeq 1/\epsilon$ take $\vx^{(j)} \defeq \phi(\vx^{(j-1)})$. If $\vx^{(j)} \notin \cX$ for some $j$, the algorithm can terminate. Otherwise, consider the distribution $\mu = \unif\{\vx_0, \dots, \vx_{M-1}\}$, and observe that
    \begin{align}
        \E_{\vx\sim\mu} [u(\phi(\vx)) - u(\vx)] &= \frac{1}{M} \sum_{j=0}^{M-1} \ab(u(\vx^{(j+1)}) - u(\vx^{(j)})) = \frac{1}{M} \ab( u(\vx^{(M)}) - u(\vx^{(0)})) \le \frac{1}{M} = \eps
    \end{align}
    as desired.
\end{proof}
A special case of this algorithm was already noted by \citet*{Zhang24:Efficient} in the context of extensive-form games and linear utilities; here, we use it for the more general concave setting where a $\log(1/\eps)$-time algorithm is actually not possible. Once again, using the above result as a drop-in replacement for the expected fixed point calculation in \citet*{Zhang25:Learning} gives the following.

\begin{theorem}\label{th:rm concave}
    Given convex set $\cX$ and function $\vm : \cX \to \R^k$, there exists a no-regret learning algorithm whose iterates are computable in time $\poly(d, k, R, 1/\eps)$, and which accepts possibly nonlinear but concave utilities, and whose $\Phi_\vm$-regret after $T$ iterations is bounded by $\eps$ after $\poly(d, k, R) / \eps^2$ rounds.
\end{theorem}

We defer the full algorithm behind~\Cref{th:rm concave} in~\Cref{appendix:shellGD}, as it largely follows~\citet*{Zhang25:Learning}.

\begin{theorem}\label{th:concave fptas eqcomp}
    There is a $\poly(d, k, R, 1/\eps)$-time algorithm for computing an $\eps$-approximate $\Phi_\vm$-equilibrium of a given concave game $\Gamma$ with a constant number of players.
\end{theorem}
\begin{remark}
    In the results of this subsection, it is crucial that the regret in \eqref{eq:regret} is defined as it is, and {\em not} of in the linearized form
    \begin{align}
        \sup_{\phi \in \Phi} \frac{1}{T} \sum_{t=1}^T \ip*{\grad u^{(t)}(\vx^{(t)}), \phi(\vx^{(t)}) - \vx^{(t)}}.
    \end{align}
    This is because of the form of the fixed point error: \Cref{prop:fptas cefp} bounds only the quantity in \eqref{eq:cefp-primal}, {\em not} the linearized fixed point error
    \begin{align}
        \E_{\vx \sim\mu} \ip{\grad u(\vx), \phi(\vx) - \vx}.
    \end{align}
    This caveat applies only to the results in this subsection, and not to \Cref{th:concave quadratic eqcomp}, because for \Cref{th:concave quadratic eqcomp} the corresponding fixed point computation (\Cref{th:gefp}) {\em does} bound the linearized fixed point error.
\end{remark}

\begin{remark}
    As we have been doing throughout the paper, our analysis here assumes that the number of players $n$ is a constant, so that expected utilities are efficiently computable. Here in the regret minimization case, since the desired dependence on $\eps$ is only inverse-polynomial, this assumption can be relaxed, because expected utilities and their gradients can be {\em estimated} to precision $\poly(d, R)/\sqrt{\ell}$ using $\ell$ samples by standard sampling techniques. For the sake of cleaniless, we do not elaborate on this construction in this paper. 
\end{remark}

\section{Conclusions and future research}

In this paper, we made significant progress on long-standing open questions regarding the complexity of correlated equilibria and swap regret minimization beyond polynomial-type games. Our \PPAD-hardness results provide evidence for the intractability of computing correlated equilibria, corroborating a long-held belief in this area. Moreover, we established unconditional, information-theoretic lower bounds ruling out the existence of strongly sublinear learners for minimizing swap regret. On the positive side, we gave a fully polynomial-time approximation scheme ($\FPTAS$) for computing $\Phi$-equilibria in general concave games, and a polynomial-time algorithm for the important special case of concave quadratic games.

Our results open several avenues for future research. The main question that we leave open concerns the complexity of correlated equilibria in structured classes such as extensive-form and Bayesian games.

Another family of open problems stems from the relationships among the contraction problems that we study. For example, what is the relationship between $\SSG$, $\ell_\infty$-contraction, $\Contr$, and $\Phi$-equilibria in concave games? Each is at least as hard as the one before it, but none are known to be polynomially equivalent. Strengthening our reduction for $\Phi$-equilibria to $\UEOPL$-hardness would be a natural next step. %

\section*{AI Disclosure}

AI assistance contributed to the proofs of~\Cref{th:fixed set hardness,th:multilinear hardness}. Concerning fixed sets, we had established \PPAD-hardness for fixed sets of constant support in low dimensions and were seeking a gap amplification argument for higher dimensions. Given this background, Astra GPT-6.0 supplied the argument using multiple scales underlying~\Cref{th:fixed set hardness}. We subsequently substantially revised and simplified the proof, extracting~\Cref{prop:non-quantitative} to explain its intuition. In developing~\Cref{th:multilinear hardness}, Astra GPT-6.0 proposed using piecewise linear interpolation and assisted with formalizing intermediate arguments. The authors wrote the manuscript and accept full responsibility for its results and their correctness.

\bibliography{dairefs}

\appendix

\section{Shell GD and the algorithm of Theorem~\ref{th:rm concave}}
\label{appendix:shellGD}

For completeness, here we provide the construction behind~\Cref{th:rm concave}, which is given in~\Cref{alg:main}. It relies on~$\shellgd$ (\Cref{sec:shellgd}), which is (projected) gradient descent but with respect to a changing constraint set. $\shellgd$ internally makes use of $\shellproj$ (\Cref{sec:shellproj}), which provides a projection oracle. In turn, $\shellproj$ uses $\shellelips$, described in~\Cref{sec:shellellips}. It strengthens the semi-separation oracle of~\citet*{Zhang25:Learning} by using \emph{concave} expected fixed points. Unlike expected fixed points, convex expected fixed points are likely hard to compute when $\epsilon$ is exponentially small (\Cref{th:fixed set hardness}), so the running time in $\shellelips$ is $\poly(1/\epsilon)$.

\subsection{Shell ellipsoid}
\label{sec:shellellips}

$\shellelips$ (\Cref{alg:shellellipsoid}) takes as input a convex set of transformations $\cF \subseteq \cB_D(\vec{0})$ for which we have efficient oracle access to and outputs \emph{either} a function $\phi \in \cF$ and an $\epsilon$-convex expected fixed point in $\Delta(\cX)$, \emph{or} a certificate---a polytope specified as the intersection of a polynomial number of halfspaces---showing that $\vol(\cF \cap \enfuns) \approx 0$.
\begin{lemma}
    \label{lemma:shellellipsoid}
    If $\cF \subseteq \cB_{D}(\vec{0})$ is a $k$-dimensional convex set that admits efficient oracle access, for $\epsilon > 0$, $\shellelips(\cF)$ (\Cref{alg:shellellipsoid}) runs in time $\poly(k, 1/\epsilon, \log D)$, and 
    \begin{itemize}
        \item either outputs $\phi \in \cF$ with an $\epsilon$-convex expected fixed point in $\cX$,
        \item or it returns a polytope $\cQ \subseteq \R^k$, expressed as the intersection of at most $\poly( k, \log(1/\epsilon), \log D)$ halfspaces, with the property that $\Phi_{\vec{m}} \subseteq \cQ$ and $\vol(\cQ \cap \cF) < \epsilon$.
    \end{itemize}
\end{lemma}

The proof of correctness is immediate following~\citet*{Daskalakis25:Efficient}.

\begin{algorithm}[!ht]
\caption{$\shellelips(\cF)$}
\label{alg:shellellipsoid}
\SetKwInOut{Input}{Input}
\SetKwInOut{Output}{Output}
\SetKw{Input}{Input:}
\SetKw{Output}{Output:}
\Input{
    \begin{itemize}[noitemsep,topsep=0pt]
        \item Oracle access to convex set $\cX \subset \R^d$
        \item Oracle access to a $k$-dimensional convex set $\cF \subseteq \cB_D(\vec{0})$
        \item Precision $\epsilon > 0$
    \end{itemize}
}
Initialize $\cE \defeq \cB_D(\vec{0})$ and $\cQ \defeq \R^k$\\
 \While{$\vol(\cE) \geq \epsilon$} {
    Set $\phi \in \cQ \cap \cF$ to be the center of the ellipsoid $\cE$\\
    Run the algorithm of~\Cref{prop:fptas cefp} with respect to $\phi$\\
    \If{it returned an $\epsilon$-convex expected fixed point $\mu \in \Delta(\cX)$ of $\phi$}{
        \textbf{return} $\phi$\\
    }
    \Else{
        Let $H$ be a halfspace that separates $\phi$ from $\enfuns$\\
        Set $\cQ \defeq \cQ \cap H$\\
    }
    Set $\cE$ to be the minimum volume ellipsoid containing $\cQ \cap \cF$
 }
 \textbf{return} $\cQ$
\end{algorithm}

\subsection{Shell gradient descent}
\label{sec:shellgd}

The basic approach consists in executing (projected) gradient descent but under a sequence of changing shell sets, $(\tilY^{(t)})_{t=1}^T$, of $\enfuns$, each containing $\enfuns$. This process is referred to as $\shellgd$ (\Cref{alg:shellgd}). If one maintains the invariance $\enfuns \subseteq \tilY^{(t)}$, it is not hard to show that $\shellgd$ minimizes external regret with respect to~$\enfuns$~\citep{Zhang25:Learning}.

\begin{lemma}[\citep{Daskalakis25:Efficient}]
    \label{lemma:shellgd}
    If the shell sets $(\tilY^{(t)})_{t=1}^T$ satisfy $\enfuns \subseteq \tilY^{(t)} \subseteq \cB_D(\vec{0})$ for any $t \in [T]$, then for any sequence of utilities $\vec{U}^{(1)}, \dots, \vec{U}^{(T)} \in [-1, 1]^k$, $\shellgd$ (\Cref{alg:shellgd}) satisfies
    \begin{equation}
        \max_{\vy^* \in \enfuns} \sum_{t=1}^T \langle \vy^* - \vy^{(t)}, \vec{U}^{(t)} \rangle \leq \frac{D^2}{2\eta} + \eta \sum_{t=1}^T \| \vec{U}^{(t)} \|_2^2.
    \end{equation}
\end{lemma}

\begin{algorithm}[!ht]
\caption{$\shellgd$~\citep{Daskalakis25:Efficient}}
\label{alg:shellgd}
\SetKwInOut{Input}{Input}
\SetKwInOut{Output}{Output}
\SetKw{Input}{Input:}
\SetKw{Output}{Output:}
\Input{Learning rate $\eta$, oracle access to convex and compact sets $\tilY^{(1)}, \dots, \tilY^{(T)} \subseteq \cB_D(\vec{0})$}\;
Initialize $\vy^{(0)} \in \tilY^{(1)}$ and $\vec{U}^{(0)} \defeq \vec{0}$\;
 \For{$ t=1, \dots, T$} {
    Update $\vy^{(t)} \defeq \Pi_{\tilY^{(t)}}( \vy^{(t-1)} + \eta \vec{U}^{(t-1)})$\\
    Output $\vy^{(t)}$ and obtain feedback $\vec{U}^{(t)} \in [-1, 1]^k$
 }
\end{algorithm}

\subsection{Shell projection}
\label{sec:shellproj}

$\shellgd$ hinges on a projection operation, which is implemented in $\shellproj$. The following lemma formalizes the guarantee of $\shellproj$.

\begin{lemma}
    \label{lemma:shellproj}
    Let $\cX$ be a convex and compact set such that $\cB_{r}(\vec{0}) \subseteq \cX \subseteq \cB_R(\vec{0})$ and $\cM$ be a convex set such that $\enfuns \subseteq \cM \subseteq \cB_D(\vec{0})$. For any $\phi \in \cB_D(\vec{0}) \subseteq \R^k$ and $\epsilon > 0$, $\shellproj$ (\Cref{alg:shellproj}) runs in time $\poly(k, 1/\epsilon, R/r, D)$ and returns
    \begin{enumerate}
        \item a shell set $\tilPhi$ satisfying $\enfuns \subseteq \tilPhi$, expressed by intersecting $\cM$ with at most $\poly(d, k, 1/\epsilon, R/r, D)$ halfspaces, and\label{item:invar}
        \item a transformation $\tilphi \in \tilPhi$ such that $\| \tilphi - \Pi_{\tilPhi}(\phi) \| \leq \epsilon$, together with an $\epsilon$-convex expected fixed point of $\tilphi$, $\mu \in \Delta(\cX)$.\label{item:proj}
    \end{enumerate}
\end{lemma}

\begin{algorithm}[!ht]
\caption{$\shellproj_{\Phi_{\vec{m}}}(\phi)$}
\label{alg:shellproj}
\SetKwInOut{Input}{Input}
\SetKwInOut{Output}{Output}
\SetKw{Input}{Input:}
\SetKw{Output}{Output:}
\Input{
    \begin{itemize}[noitemsep,topsep=0pt]
        \item Convex and compact set $\cX \subset \R^d$ with $\cB_r(\vec{0}) \subseteq \cX \subseteq \cB_R(\vec{0})$
        \item Convex set $\cM$ with $\enfuns \subseteq \cM \subseteq \cB_D(\vec{0})$
        \item Function $\phi \in \cB_D(\vec{0})$
        \item Precision $\epsilon > 0$
    \end{itemize}
}
\Output{
    \begin{itemize}[noitemsep,topsep=0pt]
        \item Convex set $\tilPhi$ such that $\enfuns \subseteq \tilPhi \subseteq \cM$
        \item Function $\tilphi \in \tilPhi$ such that $\| \tilphi - \Pi_{\tilPhi}(\phi) \| \leq \epsilon$
        \item An $\epsilon$-convex expected fixed point $\mu \in \Delta(\cX)$ of $\tilphi$
    \end{itemize}
}
Set $\epsilon'$ to be sufficiently small\\
Initialize $\tilPhi \defeq \cM$\\
 \For{$q = 0, \dots$ incremented by $\delta \defeq \nicefrac{\epsilon}{4D} $} {
    Run $\shellelips( \tilPhi \cap \cB_q(\phi))$ with precision $\vol(\cB_{\epsilon'}(\cdot))$\\
    \If{it finds $\tilphi$ with an $\epsilon$-convex expected fixed point $\mu \in \Delta(\cX)$}{
        \textbf{return} $\tilPhi, \tilphi, \mu$
    }
    \Else{
        Let $\cQ$ be the polytope returned by $\shellelips$\\
        Set $\tilPhi \defeq \tilPhi \cap \cQ$
    }
 }
\end{algorithm}

\subsection{The final algorithm}
\label{sec:put}

We are now ready to describe our algorithm for minimizing $\Phi_{\vec{m}}$-regret when $\Phi_{\vec{m}}$ is $\poly(d)$-dimensional (\Cref{alg:main}). It is based on running $\shellgd$ under the shell sets $(\tilPhi^{(t)} )_{t=1}^T$. By the guarantee of $\shellproj$, we know that $\Phi_{\vec{m}} \subseteq \tilPhi^{(t)}$ for any $t \in [T]$. Further, $\mK^{(t+1)} \in \tilPhi^{(t+1)}$ returned by $\shellproj$ in~\Cref{alg:main} is within distance $\epsilon$ of the projection required by~$\shellgd$. Applying~\Cref{lemma:shellgd}, we can bound the external regret $\Reg_{\Phi_{\vec{m}}}(T)$ of $(\mK^{(t)})_{t=1}^T$ with respect to comparators from $\enfuns$; combined with the fact that $\mu^{(t)} \in \Delta(\cX)$ is an $\epsilon$-convex expected fixed point of the function $\vx \mapsto \mK^{(t)} \vec{m}(\vx)$, it follows that the $\Phi_{\vec{m}}$-regret of \Cref{alg:main} can be upper bounded by $\Reg_{\Phi_{\vec{m}}}(T) + \epsilon T$; this is formalized below, following the argument of~\citet*{Gordon08:No}.

\begin{theorem}
    \label{theorem:main-prec}
    Let $\cX \subset \R^d$ be a convex and compact set in isotropic position for which we have a membership oracle. \Cref{alg:main} runs in time $\poly(k, T)$ and guarantees that the average $\Phi_{\vec{m}}$-regret of the learner is at most $\poly(k) / \sqrt{T}$, where $k$ is the dimension of $\Phi_{\vec{m}}$.
\end{theorem}

\begin{algorithm}[!ht]
\caption{$\Phi_{\vec{m}}$-regret minimizer}
\label{alg:main}
\SetKwInOut{Input}{Input}
\SetKwInOut{Output}{Output}
\SetKw{Input}{Input:}
\SetKw{Output}{Output:}
\Input{
    \begin{itemize}[noitemsep,topsep=0pt]
        \item Convex and compact set $\cX \subset \R^d$ in isotropic position
        \item $k$-dimensional set $\Phi_{\vec{m}}$ with respect to $\vec{m} : \cX \to \R^{k'}$, where $k = k' \cdot d$
        \item time horizon $T \in \N$
    \end{itemize}
}
\Output{An efficient $\Phi_{\vec{m}}$-regret minimizer for $\cX$}\\
Set the learning rate $\eta \propto \frac{1}{\sqrt{T}}$ and $\epsilon = \nicefrac{1}{\poly(k,T)}$ to be sufficiently small\\
Set $\mu^{(1)} \in \Delta(\cX)$ and $\mK^{(1)} \defeq \mI_{d \times k'}$ to be the identity\\
Initialize $\cM \defeq \cB_{D}(\vec{0})$ for a large enough $D \leq \poly(k)$\\
 \For{$ t=1, \dots, T$} {
    Output $\mu^{(t)} \in \Delta(\cX)$ and receive as feedback a concave utility function $u^{(t)}$\\
    Define $\R^{d \times k'} \ni \mU^{(t)} \defeq \E_{\vx \sim \mu^{(t)}} \nabla u^{(t)}( \mK^{(t)} \vec{m}(\vx))  \otimes \vec{m}(\vx)$ \\
    Set $\tilPhi^{(t+1)}, \mK^{(t+1)}, \mu^{(t+1)} \defeq \shellproj_{\Phi}(\mK^{(t)} + \eta \mU^{(t)})$ with input $\cM$ and precision $\epsilon$, where $\mu^{(t+1)} \in \Delta(\cX)$ is an $\epsilon$-convex expected fixed point of $\vx \mapsto \mK^{(t+1)} \vec{m}(\vx)$
 }
\end{algorithm}

\begin{proof}
    Fix some time $t \in [T]$. Let $u^{(t)}$ be the concave utility function, $\phi \in \Phi_{\vec{m}}$ be any comparator with $\phi = \mat{K} \vec{m} (\vx)$, $\mu^{(t)} \in \Delta(\cX)$ be the distribution played by the algorithm, and $\mK^{(t)}$ be the transformation output by $\shellproj$. We begin by bounding the term $\E_{\vx \sim \mu^{(t)}} \left[ u^{(t)}(\phi(\vx)) - u^{(t)}(\vx) \right]$, which can be written as
    \begin{equation}
        \E_{\vx \sim \mu^{(t)}} \left[ u^{(t)}(\phi(\vx)) - u^{(t)}(\mK^{(t)} \vec{m}(\vx)) \right] + \E_{\vx \sim \mu^{(t)}} \left[ u^{(t)}(\mK^{(t)} \vec{m}(\vx)) - u^{(t)}(\vx) \right]. \label{eq:regret_split}
    \end{equation}

    To bound the second term in \eqref{eq:regret_split}, we use the fact that $\mu^{(t)}$ is an $\epsilon$-convex expected fixed point of the function $ \vx \mapsto \mK^{(t)} \vec{m}(\vx)$. By definition, we have
    \begin{equation}
        \E_{\vx \sim \mu^{(t)}} \left[ u^{(t)}(\mK^{(t)} \vec{m}(\vx)) - u^{(t)}(\vx) \right] \leq \epsilon. \label{eq:fixed_point_bound}
    \end{equation}

    Next, to bound the first term in~\eqref{eq:regret_split}, we make use of the concavity of $u^{(t)}$. Specifically,
    \begin{align}
        \E_{\vx \sim \mu^{(t)}} [u^{(t)}(\phi(\vx)) - u^{(t)}(\mK^{(t)} \vec{m}(\vx)) ] &\leq \E_{\vx \sim \mu^{(t)}} \left\langle \nabla u^{(t)}(\mK^{(t)} \vec{m}(\vx)), \phi(\vx) - \mK^{(t)} \vec{m} (\vx) \right\rangle \notag \\
        &= \left\langle \E_{\vx \sim \mu^{(t)}} \nabla u^{(t)}(\mK^{(t)} \vec{m}(\vx)) \otimes \vec{m}(\vx), \mK - \mK^{(t)} \right\rangle. \label{eq:linearization}
    \end{align}
    The first term in the inner product corresponds exactly to the linear feedback vector $\vec{U}^{(t)}$ fed into $\shellgd$. As a result, combining~\eqref{eq:linearization} and \eqref{eq:fixed_point_bound} with~\eqref{eq:regret_split}, and summing over all $T$ rounds, it follows that the $\Phi_{\vec{m}}$-regret of \Cref{alg:main} can be upper bounded by $\Reg_{\Phi_{\vec{m}}}(T) + \epsilon T$, and the claim follows.
\end{proof}

\end{document}